\documentclass[11pt]{article}
\usepackage[dvips]{graphicx}
\usepackage{psfrag}
\usepackage{color}
\usepackage{latexsym,amsfonts,amssymb,amsmath,amsthm}
\usepackage{subfig}
\usepackage{tikz}
\usetikzlibrary{shapes,decorations}

\usepackage{graphics,hyperref,fullpage}  
\usepackage[numbers]{natbib} 
\hypersetup{colorlinks=true,linkcolor=blue,citecolor=blue,filecolor=blue,urlcolor=blue}
\usepackage[affil-it]{authblk}  

\newcommand{\schummer}{Schummer~\cite{Sch00} }

\newcommand{\take}[1]{\colorbox{emphcol}{$#1$}}

\newcommand{\block}[1]{
\begin{center}
   {
      \fbox{
         \begin{minipage}[t]{0.92\textwidth}
           #1
         \end{minipage}
      }
   }
\end{center}
}

\newcommand{\stpath}[4]{
	\node[circle, fill=gray!50] (s) at (0,1) {$s$};
	\node[circle, fill=gray!50] (t) at (4,1) {$t$};
	\node[circle, fill=gray!50] (u) at (2,2) {\phantom{$u$} };
	\node[circle, fill=gray!50] (d) at (2,0) {\phantom{$d$}};
	\draw[thick] (s) -- (u) node[pos=.5,sloped,above] {$#1$} -- (t) node[pos=.5,sloped,above] {$#2$} -- (d) node[pos=.5,sloped,below] {$#4$} -- (s) node[pos=.5,sloped,below] {$#3$};
}
\newcommand{\upperpath}{\draw[line width=4pt] (s) -- (u) -- (t);}
\newcommand{\lowerpath}{\draw[line width=4pt] (s) -- (d) -- (t);}

\newcommand{\stpathtriangle}[3]{
	\node[circle, fill=gray!50] (s) at (0,1) {$s$};
	\node[circle, fill=gray!50] (t) at (4,1) {$t$};
	\node[circle, fill=gray!50] (u) at (2,2) {\phantom{$u$} };
	\draw[thick] (s) -- (u) node[pos=.5,sloped,above] {$#1$} -- (t) node[pos=.5,sloped,above] {$#2$} -- (s) node[pos=.5,sloped,below] {$#3$};
}
\newcommand{\upperpathtriangle}{\draw[line width=4pt] (s) -- (u) -- (t);}
\newcommand{\lowerpathtriangle}{\draw[line width=4pt] (s) -- (t);}

\newcommand{\stpathup}[4]{
	\begin{tikzpicture}[scale=.8]
	\stpath{#1}{#2}{#3}{#4}
	\upperpath
	\end{tikzpicture}
}

\newcommand{\stpathdown}[4]{
	\begin{tikzpicture}[scale=.8]
	\stpath{#1}{#2}{#3}{#4}
	\lowerpath
	\end{tikzpicture}
}

\newcommand{\stpathtriangleup}[3]{
	\begin{tikzpicture}[scale=.8]
	\stpathtriangle{#1}{#2}{#3}
	\upperpathtriangle
	\end{tikzpicture}
}

\newcommand{\stpathtriangledown}[3]{
	\begin{tikzpicture}[scale=.8]
	\stpathtriangle{#1}{#2}{#3}
	\lowerpathtriangle
	\end{tikzpicture}
}

\newtheorem{theorem}{Theorem}
\newtheorem{corollary}[theorem]{Corollary}
\newtheorem{lemma}[theorem]{Lemma}

\newtheorem{example}[theorem]{Example}
\newtheorem{definition}[theorem]{Definition}
\newtheorem{nclaim}[theorem]{Claim} 
\newtheorem{fact}[theorem]{Fact}
\newtheorem{obs}[theorem]{Observation}

\newtheorem*{remark}{Remark}

\definecolor{emphcol}{gray}{.75}
\newcommand{\takeem}[1]{\colorbox{emphcol}}
\newcommand{\linearmechanism}[1]{$#1$-linear mechanism}

\newcommand{\influence}[1]{i\textrm{-}influence(A_{#1},\theta)}
\newcommand{\partialpay}[1]{i\textrm{-}influence(p_{#1},\theta)}

\newcommand{\sbribeproof}{strongly bribeproof}

\newcommand{\argmin}{\operatornamewithlimits{arg\ min}}

\newcommand{\SC}{\operatorname{sum}}

\renewcommand{\Re}{\mathbb{R}}

\title{\bf Bribeproof  mechanisms for two-values domains}

\author[1]{Mat\'u\v{s} Mihal\'ak}
\author[2]{Paolo Penna}
\author[2]{Peter Widmayer}

\affil[1]{Department of Knowledge Engineering, Maastricht University, The Netherlands}
\affil[2]{Department of Computer Science, ETH Zurich, Switzerland}

\begin{document}
	
\maketitle
\begin{abstract}
\schummer introduced the concept of \emph{bribeproof} mechanism which, in a context where monetary transfer between agents is possible, requires that manipulations through bribes are ruled out.
Unfortunately, in many domains, the only bribeproof mechanisms are the trivial ones which return a \emph{fixed outcome}.

This work presents one of the few  constructions of non-trivial bribeproof mechanisms for this setting. 
Though the suggested construction applies to rather restricted  domains, the results obtained are tight: for several natural problems, the method yields the only possible bribeproof mechanism and no such mechanism is possible on more general domains.
\end{abstract}

\section{Introduction}
Strategyproof mechanisms guarantee that the agents never find it convenient to misreport their types, that is, truth-telling is a dominant strategy. Such mechanisms play a key role to cope with selfish behavior, and they received a lot of attention also when considering  protocols for optimally allocating resources that necessarily involve selfish entities \cite{NisRon99}. 
One of the critical issues with strategyproof mechanisms is that agents can still  manipulate the mechanism, and improve their utilities, by \emph{bribing} one another:
\begin{quote}
	An agent can offer money to another for misreporting her type and in this way the utility of both improves.
\end{quote}
The famous second-price auction provides a clear example of such an issue. If two agents are willing to pay $10$ and $9$ for an item, the one bidding $10$ wins and pays $9$. However, if before the auction starts the winner offers some money to the other agent for bidding a low value (say $1$), then both agents would be better off (now the winner pays only $1$ and the other agent gets some money).

The concept of \emph{bribeproof} mechanism \cite{Sch00}  strengthens strategyproofness by requiring that bribing another agent is also not beneficial.
The appeal of this notion is that it does not consider unreasonably large coalitions.\footnote{The notion of \emph{coalitional strategyproofness} requires the mechanism to be immune to manipulations by any group of agents. As already observed in \cite{Sch00,Wak13}, this notion turns out to be too restrictive as it rules out all but a few unreasonable mechanisms. Moreover, large coalitions would require all members to coordinate their actions.} Despite bribeproofness is apparently adding only  a minimal condition, this has a tremendous impact on what the mechanisms can do in general:
\begin{itemize}
	\item The class of strategyproof mechanisms is extremely rich and, among others, it includes VCG mechanisms which optimize the social welfare;
	\item In contrast, the class of bribeproof mechanisms consists of only \emph{trivial} mechanisms which output a \emph{fixed outcome}  \cite{Sch00,Miz03}.
\end{itemize}
That is to say, while strategyproofness by itself is not an obstacle to optimization, the only way to get bribeproofness is to ignore the agents types, which clashes with most optimization criteria. One example of such VCG mechanisms is for the path auction problem \cite{NisRon99} where we want to select the shortest path between two nodes of a given network, every edge is owned by an agent, and the cost of the edges are private. Selecting the shortest path means that we want the solution minimizing the sum of all agents' costs, that is, to optimize the social welfare. Clearly, any trivial mechanism which returns a fixed path  has no guarantee to find the shortest path.

\subsection{Our contribution}
Because the impossibility results on bribeproof mechanisms hold for unrestricted  or for ``sufficiently rich'' domains \cite{Sch00,Miz03}, we are interested in designing bribeproof mechanisms for \emph{restricted domains}. Specifically, we present a novel construction of bribeproof mechanisms for the following class of a \emph{two-values} problems. Every feasible solution corresponds to some \emph{amount of work} allocated to each agent, and every agent has a \emph{private cost} per unit of work which is either $L$ (low) or $H$ (high).\footnote{Throughout this work we adopt the terminology used by \cite{ArcTar01} in the context of procurement auctions, though these domains have been investigated earlier in the context of allocating identical goods, as well as for certain restricted combinatorial auctions. All the results apply to these problems as well 
	(see Appendix~\ref{sec:restricted CA}).
	}
Typically 
the amount of work allocated to the agents cannot be arbitrary, but it is rather determined by the ``combinatorial structure'' of the problem under consideration. For instance, in the path auction problem \cite{NisRon99}, the mechanism must select a path in a graph, and each agent owns one edge of the graph (see  Figure~\ref{fig:two-examples-path-auctions}). Selecting a path means allocating one unit of work to each agent in the path, and no work to all other agents.

\begin{figure}
	\centering
	\subfloat[][]{%
		\begin{tikzpicture}[scale=.8]
		\node[circle, fill=gray!50] (s) at (0,1) {$s$};
		\node[circle, fill=gray!50] (t) at (4,1) {$t$};
		\draw [-,thick] (s) .. controls (2,1.7) ..
		(t) node[pos=.5,sloped,above] {$\theta_1$} ;
		\draw [-,thick] (s) .. controls (2,0.3) ..
		(t) node[pos=.5,sloped,below] {$\theta_2$} ;
		\end{tikzpicture}
		\label{fig:parallel-links}}
	\hspace{2cm}
	\subfloat[][]{%
		\begin{tikzpicture}[scale=.8]
		\stpath{\theta_1}{\theta_2}{\theta_3}{\theta_4}
		\end{tikzpicture}
		\label{fig:sp-diamond}
	}
	\caption{Two instances of the path auction problem.}
	\label{fig:two-examples-path-auctions}
\end{figure}
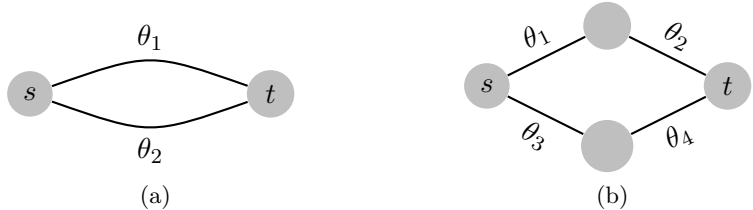

\noindent
In a nutshell our results can be summarized as follows:
\begin{itemize}
	\item An extremely simple construction yields bribeproof mechanisms if the underlying algorithm satisfies certain monotonicity conditions (Section~\ref{sec:class-mechanisms}). 
	\item One application of the above result, is a class of bribeproof mechanisms optimizing the social welfare for  every \emph{binary} allocation problem, that is, whenever each agent is either selected or not selected (Section~\ref{sec:binary}). 
	\item These mechanisms actually \emph{characterize} the whole class of bribeproof mechanisms for certain problems, including the path auction one, and the boundary conditions for which such mechanisms exist (Section~\ref{sec:characterizations}).
	\item The positive result is more general as it can be applied to non-binary problems and to other optimization criteria (Section~\ref{sec:non-binary}). 
\end{itemize}\newcommand{\money}{\mathcal{M}}
More in detail, our mechanisms simply provide all agents the \emph{same} amount of money $\money$ for each unit of work that they get allocated (Definition~\ref{def:linear-mechanism}). Such mechanisms are bribeproof if certain monotonicity conditions hold (Theorem~\ref{th:wokload:bribeproof}).
Roughly speaking, these conditions relate the ``influence'' that an agent has on her own allocation  to the influence she has on the \emph{others'} allocation. In particular, by taking the special case $\money=\frac{L+H}{2}$ in our construction leads to the following natural sufficient condition (Corollary~\ref{cor:wokload:bribeproof:1/2}):
\begin{quote}
\emph{Bounded influence:} No agent can change the allocation of \emph{another} agent by more than the change caused to \emph{her own} allocation.
\end{quote}
For the class of \emph{binary allocations}, where the allocated work is $zero$ or $one$, 
this condition is nothing but \emph{non-bossiness}: no agent can change the allocation of the others, without changing her own allocation (Theorem~\ref{th:sbribeproof:equivalence:sbribeproof}). The main positive result here is that every problem in which one wants to minimize the weighted \emph{sum} of all agents' costs admits an exact strongly bribeproof mechanism (Theorem~\ref{th:binary-utilitarian}). 
Interestingly, our general construction provides both characterizations of bribeproof mechanisms as well as the boundary conditions for which such mechanisms exist in several problems:
\begin{itemize}
	\item For the path auction problem, our mechanism with $\money=\frac{L+H}{2}$ is essentially the only possible, and no mechanism exist on slightly more general domains (with three values, or heterogeneous two values), nor collusion-proof mechanisms for coalitions of three or more agents.
	\item For the $k$-items procurement auction, the mechanism with $\money=M$ is bribeproof on \emph{three values} domains $L$ (low), $M$(medium), $H$(high). This is the only mechanism for $k=1$ and no mechanism for \emph{four values} domains exist. 
\end{itemize}
We then turn our attention to problems with different objective function and non-binary allocations. Specifically, we consider minimizing the \emph{maximum} cost among the agents (note that this is different from welfare maximization which would minimize the \emph{sum} of all agents' costs). In the scheduling terminology, we aim at minimizing the \emph{makespan} on related machines \cite{ArcTar01}. In the fractional version, when each job can be divided among the machines, we get an exact bribeproof mechanism (Theorem~\ref{th:fractional-makespan}) since the problem is equivalent to allocating a single job (in a fractional way) and the bounded influence condition holds. On the contrary, when jobs cannot be divided \cite{ArcTar01}, we show that our method cannot give exact or even approximate bribeproof mechanisms. The existence of other mechanisms for this and other problems is an interesting open question. More in general, it would be interesting to obtain approximate mechanisms when the domain does not allow for exact ones.

\subsection{Related work}\label{sec:previous}
 \schummer introduced the notion of bribeproofness and proved that, on certain domains, the only bribeproof mechanisms are the \emph{trivial} mechanisms which return a \emph{fixed outcome}; \cite{Miz03} proved the same but under weaker assumptions. In simpler domains, bribeproof (or even collusion-proof) mechanisms can be obtained via \emph{take-it-or-leave-it}  prices \cite{GolHar05,GolVen13}: these mechanisms fix a price for each agent, who then  wins a copy of the item if bidding above this price, independently of what happens to the other agents.
 Note that our mechanisms are different from these mechanisms since in our setting we cannot treat  agents separately.

 Though strategyproofness is much less stringent and quite well understood, restricting the domain is also very common, as unrestricted domains are often unrealistic and impose unnecessary limitations (see e.g. \cite{Rob79,LavMuaNis03,DobNis15}).  In multi-dimensional
domains, minimizing the \emph{makespan} or \emph{min-max fairness} is not possible using strategyproof mechanisms  \cite{NisRon99,Gam07,KouVid07,MuaSch07,LavSwa09,AshItaDobLav12}, while for one-parameter domains  optimal solutions are possible \cite{ArcTar01,MuaSch07} also in polynomial time \cite{ChrKov08}. Our domains are at the intersection of one-parameter domains in \cite{ArcTar01} and the two-values domains in \cite{LavSwa09}, and they also appear in study of revenue of take-it-or-leave-it identical items auctions \cite{GolVen13}.  One-parameter domains have been  studied by \cite{Mye81} who characterized strategyproofness and obtained optimal-revenue mechanisms for selling a single item.

The strong limitations imposed by bribeproofness lead to the study of weaker or variants of this notion. \emph{Group strategyproofness}  assumes that the members of the coalitions can coordinate their reports but cannot exchange compensations (see e.g. \cite{Mou99,PouVid12,Muk12,Jua13}). The restriction to coalitions of size two is called \emph{pairwise strategyproofness} \cite{serizawa2006pairwise}, and it corresponds to strong bribeproofness when compensations between agents are not allowed.  
The class of \emph{deferred acceptance} mechanisms \cite{MilSeg14} satisfies (weakly) group strategyproofness\footnote{This condition relaxes group strategyproofness, by requiring that no coalition could deviate from truth-telling in a way that makes
	\emph{all} of its members strictly better off.}, at the price of significantly worse social welfare even in rather simple settings \cite{DutGkaRou14}.
 Mechanisms with \emph{verification}  \cite{PenVen12} are based on the assumption that it is possible to partially verify the agents types after the solution is computed. \emph{Collusive dominant-strategy truthfulness}  \cite{CheMic12} is based on the idea that the mechanism asks the agents to report also their coalitions, and it provides better performance for selling identical items. In the so-called \emph{single-peaked} domains, agents receive a variable amount of a divisible item, and they can bribe each other by transferring part of the item (with no money involved).  
  Interestingly enough, \cite{Wak13} characterizes bribeproofness in terms of a \emph{bounded impact} condition which is very similar to our bounded influence, despite the two settings being not equivalent.  
Finally, while \cite{SchGEB00} shows that bribeproofness is closely related to Pareto efficiency together with strategyproofness, 
\cite{ohseto2000strategy} proved that the latter two requirements cannot be achieved simultaneously in a setting involving finite ``small'' domains.

\subsection{Preliminaries}\label{sec:preliminaries}
There is a set $N=\{1,2,\ldots,n\}$ of $n \geq 2$ agents and a set $\mathcal{A}\subseteq \Re_+^n$ of feasible allocations, where each allocation $a\in \mathcal{A}$ is an $n$-dimensional vector $(a_1,\ldots,a_n)$ with $a_i$ being the amount of
work allocated to agent $i$. For each agent $i$, her cost for an allocation $a$ is
equal to 
\[
  a_i \cdot \theta_i,
\]
where $\theta_i\in \Re$ is some private number called the \emph{type} of this agent (her cost for a unit of work). 
Every type $\theta_i$ belongs to a publicly-known set $\Theta_i$ which is the domain of agent
$i$, and the agent can misreport her type $\theta_i$  to any $\hat \theta_i\in \Theta_i$. The cross-product $\Theta := \Theta_1\times\cdots\times \Theta_n$ is the types domain representing the possible type vectors that can be reported by the agents. 

A mechanism is a pair $(A,p)$ where $A:\Theta \rightarrow \mathcal{A}$ is an algorithm and $p:\Theta \rightarrow \Re^n$ is a suitable payment function. For any type vector $\hat \theta \in \Theta$  reported by the agents,  each agent $i$ receives $p_i(\hat \theta)$ units of money and $A_i(\hat \theta)$ units of work. 
 A mechanism $(A,p)$ is \textbf{bribeproof} if for all $\theta \in \Theta$,  all $i$ and $j$, and
   all $\hat \theta_i \in \Theta_i$ 
  \begin{align}
   & p_i(\theta) -  A_i(\theta)\cdot \theta_i &+ & & p_j(\theta) -  A_j(\theta)\cdot \theta_j & &\geq   \label{eq:bribeproof:breibee} \\ 
  & p_i(\hat \theta) -  A_i(\hat \theta)\cdot \theta_i &+ & & p_j(\hat \theta) -  A_j(\hat \theta)\cdot \theta_j & \nonumber
  \end{align}
where 
$
 \hat \theta = (\hat \theta_i,\theta_{-i}):=(\theta_1,\ldots,\theta_{i-1},\hat \theta_i, \theta_{i+1},\ldots,\theta_n) 
$
denotes the vector obtained by replacing the $i^{th}$ entry of $\theta$ with $\hat  \theta_i$. 
  
Inequality \eqref{eq:bribeproof:breibee} says that no
agent $j$ can bribe another agent $i$ with $b$ units of money to misreport her type so that  they both improve. 
By taking $i=j$ in the definition above, we obtain the
(weaker) 
notion of \textbf{strategyproof} mechanism, that is,  
$ p_i(\theta) -  A_i(\theta)\cdot \theta_i  \geq  p_i(\hat \theta_i,\theta_{-i}) -  A_i(\hat \theta_i,\theta_{-i})\cdot \theta_i,$ 
for all $\theta\in \Theta$, for all $i$, and for all $\hat \theta_i \in \Theta_i$.
Strong bribeproofness requires that no two agents can improve even if they \emph{jointly} misreporting their types (see \cite[page 184]{Sch00}). Let $(\hat \theta_i,\hat \theta_j,\theta_{-ij})$ denote the vector obtained by replacing the $i^{th}$ and the $j^{th}$ entry of $\theta$ with $\hat  \theta_i$ and $\hat \theta_j$, respectively.  A mechanism $(A,p)$ is \textbf{\sbribeproof} if inequality \eqref{eq:bribeproof:breibee} holds also for all 
$
\hat \theta = (\hat \theta_i,\hat \theta_j,\theta_{-ij})
$,
with  $\theta_i\in \Theta_i$, $\theta_j \in \Theta_j$, and $\theta\in \Theta$. 

A domain $\Theta$ is a \textbf{two-values domain} if there exist two constants $L$ and $H$ with $L<H$ such that 
 $
  \Theta_i = \{L,H\}
$
for all $i\in N$. More generally, for any ordered sequence of reals $w_1 < w_2 < \cdots < w_k$, we denote by $\Theta^{(w_1,w_2,\ldots,w_k)}$ the \textbf{$\mathbf{k}$-values} domain $\Theta$ such that
$
 \Theta_i = \{w_1,w_2,\ldots,w_k\}
$
for all $i\in N$. 
 We say that a mechanism is (strongly) bribeproof over a $k$-values domain if 
the corresponding condition \eqref{eq:bribeproof:breibee} holds for $\Theta$ being a $k$-values domain. 

\begin{example}[path auction and perfectly divisible good]\label{exa:path-auction}\label{exa:perfectly-divisible}
	In the \emph{path auction problem} instance in Figure~\ref{fig:two-examples-path-auctions}, the two feasible  allocations are  
	$(1,1,0,0)$ for the ``upper path'' and $(0,0,1,1)$ for the ``lower path''. The problem of allocating a single \emph{perfectly divisible} good among the agents corresponds to the set of feasible allocations consists of all vectors $a=(a_1,\ldots,a_n)$ such that 
	$
	a_i \geq 0 $ and $\sum_{i=1}^n a_i = 1$.
\end{example}

In a bribeproof mechanism the payments must depend only on the allocation:
\begin{fact}\label{fac:adj:pay}
	 We say that two type vectors $\theta'$ and $\theta''$ are $A$-equivalent 
	 if they differ in exactly one agent's type
	 and algorithm $A$ returns the same allocation in both cases. That is, 
	 $
	 	\theta''=(\theta''_i, \theta'_{-i})$ and $A(\theta'')=A(\theta').
	 $
  In a bribeproof mechanism $(A,p)$ the payment for two $A$-equivalent type vectors must be the
  same.
\end{fact}
\begin{proof}
Suppose by way of contradiction that $p(\theta')\neq p(\theta'')$ and, without loss of generality, that  $p_j(\theta')>p_j(\theta'')$ for some agent $j$. Since $A(\theta')=A(\theta'')$, this violates bribeproofness \eqref{eq:bribeproof:breibee}. 
\end{proof}

\section{A class of bribeproof mechanisms}\label{sec:class-mechanisms}
The idea to obtain bribeproof mechanisms is to pay each agent
the same fixed amount for each unit of work she gets allocated.

\begin{definition}[linear mechanism]\label{def:linear-mechanism}
  A mechanism $(A,p)$ is a \linearmechanism{\lambda} if every agent $i$ receives a fixed payment $f_i$ plus $\lambda L + (1-\lambda)H$
   units of money for each unit of allocated work, where $\lambda \in [0,1]$. That is,
  \[
    p_i(\theta) = A_i(\theta) \cdot q^{(\lambda)} + f_i\ \ \ \ \mbox{ where } \ \ \ \ q^{(\lambda)} =\lambda L + (1-\lambda)H
  \]
  for all $i$ and for all $\theta \in \Theta$.
\end{definition}

\begin{remark}
Note that, in Definition~\ref{def:linear-mechanism},  we limit ourself to $q \in [L,H]$ because otherwise the mechanism would not be even strategyproof in general. Moreover, the constants $f_i$ can be used to rescale the payments without affecting bribeproofness. For instance, one can set each $f_i$ so that truthfully reporting agents are guaranteed a nonnegative utility, i.e., the mechanism satisfies voluntary participation or individual rationality.   
\end{remark}

In the following we define
\[
 \influence{k}:= A_k(L,\theta_{-i}) - A_k(H,\theta_{-i}).
\]
The  monotonicity condition for
strategyproofness \cite{ArcTar01,Mye81}  requires that the allocation of each agent is weakly decreasing
in her reported cost, that is, for all $i\in N$ and for all $\theta\in \Theta$
\begin{align}
 \label{eq:mon}
 \influence{i}&\geq 0 & \mbox{(monotonicity)}&.
\end{align}
We next show that a stronger condition suffices for bribeproofness.

\begin{theorem}\label{th:wokload:bribeproof}
  The \linearmechanism{\lambda} is bribeproof for a two-values domain if and only if algorithm $A$ satisfies the following conditions: for all $\theta\in \Theta^{(L,H)}$ and for all $i\in N$ condition \eqref{eq:mon} holds and, for all $\ell\in N$ with $\theta_\ell=L$ and for all $h\in N$ with $\theta_h=H$, 
  \begin{align}
   (1-\lambda)&\cdot\influence{i} & \geq & &(\lambda - 1)&\cdot\influence{\ell}, \label{eq:pmon:LL}\\ 
   (1-\lambda)&\cdot\influence{i} & \geq &  &\lambda\phantom{-1}&\cdot\influence{h}, \label{eq:pmon:LH}\\
   \lambda\phantom{-1}&\cdot\influence{i}     & \geq & &(1- \lambda)&\cdot\influence{\ell}, \label{eq:pmon:HL}\\
   \lambda\phantom{-1}&\cdot\influence{i}     & \geq & &-\lambda\phantom{-1}    &\cdot\influence{h}. \label{eq:pmon:HH}
  \end{align}
\end{theorem}

\begin{proof}
	It is easy to see that condition \eqref{eq:bribeproof:breibee} for $i=j$ is equivalent to the monotonicity condition \eqref{eq:mon}.\footnote{This case corresponds to the notion of strategyproofness and the equivalence to monotonicity is proved in \cite{Mye81,ArcTar01} for the more general class of one-parameter domains.} We thus consider the case $i\neq j$ and show that \eqref{eq:bribeproof:breibee} is equivalent to the conditions \eqref{eq:pmon:LL}-\eqref{eq:pmon:HH} above. 
	By definition of \linearmechanism{\lambda} the utility of a generic  agent $k$ is 
	\begin{align}
		p_k(\hat \theta) - A_k(\hat \theta) \cdot \theta_k  &=  f_k + A_k(\hat \theta) \cdot (\lambda L + (1-\lambda)H - \theta_k) \nonumber 
		\intertext{which can be rewritten as follows using $\theta_k\in \{L,H\}$:}  
		p_k(\hat \theta) - A_k(\hat \theta) \cdot \theta_k  &= f_k + A_k(\hat \theta)\cdot (H-L) \cdot
		\begin{cases}
			1-\lambda & \mbox{if  $\theta_k=L$}; \\
			\phantom{1} -\lambda  & \mbox{if  $\theta_k=H$}.
		\end{cases}
		\label{eq:linear-mechanism:utility}
	\end{align}
	We rewrite \eqref{eq:bribeproof:breibee} according to the above utility function for each of the four possible cases of $\theta_i$ and $\theta_j$:
	\begin{description}
		\item[($\theta_i=L$ and $\theta_j=L$)] This case corresponds to 
		\[
		(1-\lambda)A_i(L,\theta_{-i}) + (1-\lambda)A_j(L,\theta_{-i}) \geq  (1-\lambda)A_i(H,\theta_{-i}) + (1-\lambda)A_j(H,\theta_{-i})
		\]
		which is equivalent to  \eqref{eq:pmon:LL}.
		\item[($\theta_i=L$ and $ \theta_j=H$)] This case corresponds to 
		\[
		(1-\lambda)A_i(L,\theta_{-i}) + (-\lambda)A_j(L,\theta_{-i}) \geq  (1-\lambda)A_i(H,\theta_{-i}) + (-\lambda)A_j(H,\theta_{-i})
		\]
		which is equivalent to  \eqref{eq:pmon:LH}.
		\item[($\theta_i=H$ and $ \theta_j=L$)] This case corresponds to 
		\[
		(-\lambda)A_i(H,\theta_{-i}) + (1-\lambda)A_j(H,\theta_{-i}) \geq  (-\lambda)A_i(L,\theta_{-i}) + (1-\lambda)A_j(L,\theta_{-i})
		\]
		which is equivalent to  \eqref{eq:pmon:HL}.
		\item[($\theta_i=H$ and $ \theta_j=H$)] This case corresponds to 
		\[
		(-\lambda)A_i(H,\theta_{-i}) + (-\lambda)A_j(H,\theta_{-i}) \geq  (-\lambda)A_i(L,\theta_{-i}) + (-\lambda)A_j(L,\theta_{-i})
		\]
		which is equivalent to  \eqref{eq:pmon:HH}.
	\end{description}
	This completes the proof.
\end{proof}

A simple corollary 
of the previous theorem is that the following natural condition implies bribeproofness when setting $\lambda=1/2$.
\begin{definition}[bounded influence]
 An algorithm $A$ satisfies bounded influence if, for all $\theta\in \Theta$ and for all $i,j\in N$, the following condition holds:
 \begin{equation}
  \label{eq:bounded-influence}
  \influence{i} \geq |\influence{j}|.
 \end{equation}
\end{definition}

\begin{corollary}\label{cor:wokload:bribeproof:1/2}
 The \linearmechanism{\left(\frac{1}{2}\right)}  is bribeproof for two-values domains if and only if its algorithm $A$ satisfies bounded influence.
\end{corollary}

\section{Binary allocations}\label{sec:binary}
In this section we apply our results to the case of \emph{binary allocations}, that is, the problems in which each agent is allocated either an amount equal $0$ or $1$   (the path auction and the $k$-item procurement auction are two examples).

We first observe that bounded influence boils down to the following natural condition called \emph{non-bossiness} (no agent can change the allocation of another agent without changing her own allocation), and for our construction this condition is equivalent to (strong) bribeproofness.

\begin{definition}[non-bossiness]\label{def:stable}
  An algorithm $A$ satisfies non-bossiness if, for all $i$ and for all $\theta$, the following implication holds:
  if $\influence{i}=0$ then $\influence{j}=0$ for all $j$.
\end{definition}

\begin{theorem}\label{th:sbribeproof:equivalence}
 For binary allocations and two-values domains, the following statements are equivalent:
 \begin{enumerate}
  \item The \linearmechanism{\left(\frac{1}{2}\right)} $(A,p)$ is bribeproof. \label{th:sbribeproof:equivalence:bribeproof}
  \item Algorithm $A$ satisfies monotonicity  and non-bossiness. \label{th:sbribeproof:equivalence:non-bossiness}
  \item The \linearmechanism{\left(\frac{1}{2}\right)} $(A,p)$ is \sbribeproof. \label{th:sbribeproof:equivalence:sbribeproof}
 \end{enumerate}
 \end{theorem}

Note that, in general, there exist bribeproof mechanisms whose algorithm does not satisfy non-bossiness.
\begin{example}[non-bossiness is not necessary]\label{exe:non-boossy-not-necessary} For two agents, consider the algorithm which selects the one having a strictly smaller type, if that is the case, and otherwise no agent is selected:\footnote{This example applies to the problem of allocating a single item considered in \cite{MisQua13} where not allocating the item is also an option.}
	\begin{align*}
		A(L,H)&= (1,0), & A(H,L)&=(0,1) & \mbox{and}& &  A(L,L)&=(0,0)=A(H,H).
	\end{align*}
	Note that this algorithm does not satisfies non-bossiness. We next show that the \linearmechanism{1} is \sbribeproof. Its payments are
	\begin{align*}
		p(L,H)&= (L,0), & p(H,L)&=(0,L) & \mbox{and}& &  p(L,L)&=(0,0)=p(H,H),
	\end{align*}
	and it is easy to check that a truthful report yields utility $0$ to both agents, while any misreport can only give the same or a negative utility.
\end{example}

\subsection{Utilitarian (min-sum) problems}
The main application of Theorem~\ref{th:sbribeproof:equivalence} is a general construction of exact mechanisms for utilitarian problems (see e.g. \cite{NisRon99}), that is, for minimizing the (weighted) \emph{sum} of all agents' costs.

\begin{definition}[weighted social cost minimization]
 An algorithm $A$ minimizes the weighted social cost if there exist  nonnegative constants $\{\alpha_i\}_{i\in N}$   and arbitrary constants $\{\beta_a\}_{a \in \mathcal A}$ such that, for all $\theta\in \Theta$, it holds that
\[
 A(\theta) \in \argmin_{a\in \mathcal{A}} \{\SC(a,\theta)\},
\]
where $\SC()$ is defined as 
 $\SC(a,\theta) := \left(\sum_{i\in N} \alpha_i  a_i   \theta_i\right) + \beta_a.$
\end{definition}
Obviously, every mechanism minimizing the sum of all agents' costs correspond to the case $\alpha_i=1$ and $\beta_a=0$. Welfare maximization problems correspond to the case in which agents have valuations instead of costs.

\begin{definition}[consistent ties]
 An algorithm $A$ minimizes the weighted social cost breaking ties consistently if there exists 
a total order $\preceq$ over the set $\mathcal{A}$ of feasible allocations such that, for all $\theta\in \Theta$ and for all $a'\in \mathcal{A}$,  the following implication holds:
if $\SC(A(\theta),\theta)=\SC(a',\theta)$, then  $A(\theta) \preceq a'.$
\end{definition}

\begin{theorem}\label{th:binary-utilitarian}
 For binary allocation problems over two-values domains, if algorithm $A$ minimizes the weighted social cost  breaking ties consistently, then the corresponding \linearmechanism{\left(\frac{1}{2}\right)} is \sbribeproof.
\end{theorem}

\begin{proof}
 We show that  every algorithm $A$ which minimizes the weighted social cost  breaking ties consistently, satisfies non-bossiness and monotonicity \eqref{eq:mon}. The theorem then follows from Theorem~\ref{th:sbribeproof:equivalence}.

 It is convenient to rewrite the weighted social cost into two parts, the contribution of
a fixed agent $i$ and the rest:
\begin{align*}
 \alpha_i a_i \theta_i &+  \SC_{-i}(a, \theta_{-i}) & \mbox{ where }& & \SC_{-i}(a, \theta_{-i})&:=
  (\sum_{j\in N \setminus  \{i\}} \alpha_j a_j  \theta_j)+ \beta_a.
\end{align*}
For ease of notation, also let $\theta^L:= (L,\theta_{-i})$ and $\theta^H:=(H,\theta_{-i})$. Observe that since $A$ minimizes the weighted social cost we have the following:
\begin{align}
 \SC(A(\theta^L),\theta^L)& &=&  &L \alpha_i A_i(\theta^L)  + \SC_{-i}(A(\theta^L),\theta_{-i}) & &\leq & \label{eq:utilitarian:SC-L} \\ 
 \SC(A(\theta^H),\theta^L)& &=&  &L \alpha_i A_i(\theta^H)  + \SC_{-i}(A(\theta^H),\theta_{-i}),  & &  \nonumber \text{ and }\\
 \SC(A(\theta^H), \theta^H)& &=& &H \alpha_i A_i(\theta^H)  + \SC_{-i}(A(\theta^H),\theta_{-i}) & &\leq \label{eq:utilitarian:SC-H} \\ \nonumber
 \SC(A(\theta^L), \theta^H)& &=& &H \alpha_i A_i(\theta^L) + \SC_{-i}(A(\theta^L),\theta_{-i}) .& & 
 \end{align}
First, we show the following implication:
\begin{align}\label{eq:utilitarian:SC-non-bossiness}
\alpha_iA_i(\theta^H)&=\alpha_iA_i(\theta^L) & \Rightarrow&  & A(\theta^L)&=A(\theta^H).
\end{align}
The left-hand side implies that both inequalities \eqref{eq:utilitarian:SC-L} and \eqref{eq:utilitarian:SC-H} hold with ``$=$''. Since ties are broken consistently, we have $A(\theta^L) \preceq A(\theta^H)$ by \eqref{eq:utilitarian:SC-L} and $A(\theta^H) \preceq A(\theta^L)$ by \eqref{eq:utilitarian:SC-H}, thus implying $A(\theta^L)=A(\theta^H)$.

Now observe that \eqref{eq:utilitarian:SC-non-bossiness} implies that $A$ satisfies non-bossiness, and thus it only remains to prove that the monotonicity condition holds. By summing inequalities \eqref{eq:utilitarian:SC-L} and \eqref{eq:utilitarian:SC-H} we obtain $\alpha_i(H-L) A_i(\theta^H) \leq \alpha_i(H-L) A_i(\theta^L)$. From this the inequality
\[
A_i(\theta^H) \leq A_i(\theta^L)
\]
follows immediately for the case $\alpha_i>0$, while for $\alpha_i=0$ it follows by \eqref{eq:utilitarian:SC-non-bossiness}. By definition, the inequality above is equivalent to the monotonicity condition \eqref{eq:mon}.  
\end{proof}

The mechanism for the path auction problem consists in paying each agent in the chosen path an amount equal to $\money=\frac{L+H}{2}$, where the algorithm breaks ties between paths in a fixed order. Similar mechanisms can be obtained for other utilitarian problems like minimum spanning tree \cite{NisRon99} or for the  $k$-item procurement auction (see Section~\ref{sec:characterizations} for the latter).

\section{Characterizations for two problems}\label{sec:characterizations}
In this section we show that the  \linearmechanism{\left(\frac{1}{2}\right)} is the only bribeproof mechanism  for the path auction on general networks (this result applies also to combinatorial auctions with known single minded bidders -- see Appendix~\ref{sec:restricted CA} for details). We then obtain analogous characterizations for the $k$-item procurement auction in terms of our \linearmechanism{\lambda}s. 

\subsection{Path auction mechanisms}

As for the path auction problem, we actually prove a stronger result saying that the \linearmechanism{\left(\frac{1}{2}\right)} is the only bribeproof for 
the simple network in Figure~\ref{fig:sp-diamond} on the following generalization of two-values domains:

\begin{definition}
The path auction with $\epsilon$-perturbed domain ($\epsilon \geq 0$) is the path auction problem restricted to the network in Figure~\ref{fig:sp-diamond} in which the agents domain are as follows: 
$\Theta_1 =\Theta_2=\{L-\epsilon,H+\epsilon\}$ and $\Theta_3 =\Theta_4=\{L,H\}.$
\end{definition}
Clearly the two-values domain corresponds to setting $\epsilon=0$.

\begin{theorem}\label{th:sp:generalized-domain}
A mechanism which is bribeproof for the path auction with $\epsilon$-perturbed domain  must be 
a \linearmechanism{\left(\frac{1}{2}\right)}.
\end{theorem}

\begin{proof}[Main Ideas]
	We first show that, no matter how the mechanism breaks ties, the payments must depend only on which path is selected (using Fact~\ref{fac:adj:pay}). 
	This means that the payments are of the form
	\[
	p_i(\theta) = f_i + 
	\begin{cases}
	q_i & \mbox{if $i$ is selected for types $\theta$,} \\
	0  & \mbox{otherwise.}
	\end{cases}
	\]    
	In order to conclude that the mechanism must be a \linearmechanism{\left(\frac{1}{2}\right)} it is enough to prove that  $q_i=\frac{L+H}{2}$ for all $i$. This is the technically involved part, because we have to consider the possible tie breaking rules. At an intermediate step, we show that $q_1+q_2=L+H=q_3+q_4$, for otherwise there exists a coalition which violates bribeproofness. 
\end{proof}

By taking $\epsilon=0$ we obtain a characterization for this problem:

\begin{corollary}\label{cor:sp:characterization}
 The \linearmechanism{\left(\frac{1}{2}\right)} is the only bribeproof mechanism for the path auction on general networks.
\end{corollary}

Since in these instances of path-auction problem \linearmechanism{\left(\frac{1}{2}\right)}  are \emph{not} bribeproof on three-values domains, we obtain the following result.

\begin{theorem}\label{th:sp:no-three-vals}
	There is no bribeproof mechanism for the path auction problem on general networks and for 
	three-values domains.
\end{theorem}

Theorem~\ref{th:sp:generalized-domain} implies that we cannot extend the positive result to coalitions of larger size, nor to \emph{heterogeneous} two-values domains in which $\theta_i\in\{L_i,H_i\}$. 

\begin{corollary}\label{cor:path-auction:no-collusion-proof}
 There is no collusion-proof mechanism for the path auction problem 
 on  general networks and two-values domains. The same remains true even if we restrict to coalitions of size three (in which two agents bribe another for misreporting her type).
\end{corollary}

\begin{corollary}\label{cor:sp:no-heterogeneous}
 There is no bribeproof mechanism for the path auction problem on general networks and certain heterogeneous two-values domains.
\end{corollary}

\subsection{$k$-Item auction mechanisms}
We remark that on a simple network consisting of $n$ parallel edges the path auction problem is the same as the $1$-item procurement auction. 
\begin{figure}
	\block{\textbf{\linearmechanism{\lambda_M} (normalized to $f_i=0$):}
		Select the $k$ agents with smallest types, breaking ties in favor of agents with smaller index; Pay each of the selected agents an amount $M$, and non-selected agents receive no money.               
	}
	\caption{A bribeproof mechanism for $k$-item procurement auction over three-values domains $\Theta^{(L,M,H)}$.}
	\label{fig:M-compensation}
\end{figure}
For the $k$-item procurement auction over three values domains, we consider the mechanism which provides a payment equal $M$ to the selected agents (see Figure~\ref{fig:M-compensation}). Note that this is a \linearmechanism{\lambda} with $$\lambda=\lambda_M:= \frac{H-M}{H-L}$$
which is the value such that $$M=L\lambda_M+(1-\lambda_M)H.$$ 
We show that the \linearmechanism{\lambda_M} in Figure~\ref{fig:M-compensation} is bribeproof over 
three-values domains (Theorem~\ref{th:one-job:bribeproof}) and, for the case $k=1$,   only \linearmechanism{\lambda_M}s can be bribeproof (Theorem~\ref{th:parallel:three-vals:char}). 
\begin{theorem}\label{th:one-job:bribeproof}
	The  \linearmechanism{\lambda_M} is bribeproof for the $k$-item procurement auction in the case of three-values domains. 
\end{theorem}

Also in this problem our construction yields the only mechanism, and results cannot be extended to more complex domains.

\begin{theorem}\label{th:parallel:three-vals:char}
	The \linearmechanism{\lambda_M} is the only bribeproof mechanism for the $1$-item procurement auction with three-values domains and two agents.
\end{theorem}

This implies the impossibility result.
\begin{corollary}\label{cor:parallel:four-vals:char}
	There is no bribeproof mechanism for the $1$-item procurement auction with two agents and four-values domains. 
\end{corollary}

\subsection{Back to path auction: Graph structure and three-values domains}
One way to restate the result on $1$-item auction, is that path auction admits a bribeproof mechanism on three-values domains when restricted to the \emph{parallel links} graph in Figure~\ref{fig:parallel-links}. The proof of Theorem~\ref{th:sp:no-three-vals} says that on the graph in Figure~\ref{fig:sp-diamond} there are no mechanisms for three-values domains. We next strengthen the result to the simple ``triangle'' graph in Figure~\ref{fig:sp:no-three-vals-sp:triangle}. Unlike Theorem~\ref{th:sp:no-three-vals}, this applies to some combination of values defining the domain. The result says that parallel links is the ``most general'' 
graph for which bribeproof mechanisms on any three-values domain exist.

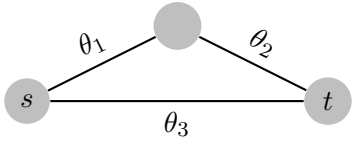
\begin{figure}
	\begin{center}
		\begin{tikzpicture}
		\stpathtriangle{\theta_1}{\theta_2}{\theta_3}
		\end{tikzpicture}
	\end{center}
	\caption{A simple network for which there is no bribeproof mechanism for the path auction problem on certain three values domains.}
	\label{fig:sp:no-three-vals-sp:triangle}
\end{figure}

\begin{theorem}\label{th:sp:no-three-vals:triangle}
	There is no bribeproof mechanism for the path auction problem on the network in Figure~\ref{fig:sp:no-three-vals-sp:triangle} and for 
	some three-values domains. In particular, this holds whenever $\Theta^{(L,M,H)}$ satisfies 
	\begin{align*}
		2L &< M, &  L+H &< 2M, & \mbox{and}& & L+M &<
		H.
	\end{align*}
\end{theorem}

\section{Min-max fairness and non-binary  problems}\label{sec:non-binary}

In this section we consider problems  with \emph{min-max} fairness optimization criteria, and \emph{non-binary} allocations. Thus, the algorithm 
$A$ should satisfy 
\begin{equation}
A(\theta) \in \argmin_{a \in \mathcal{A}} \max_{i\in N} \{a_i \cdot \theta_i\}.	\label{eq:min-max}
\end{equation}
In particular we consider the problem of allocating a perfectly divisible item (Example~\ref{exa:perfectly-divisible}) according to the above min-max fairness criteria \eqref{eq:min-max}. In such allocation  all agents will get some positive amount so that all costs will be identical.

\begin{theorem}[min-max fairness]\label{th:fractional-makespan}
	There is a \sbribeproof\  \linearmechanism{\left(\frac{1}{2}\right)} satisfying min-max fairness for allocating a perfectly divisible item. 
\end{theorem}

We next consider the problem of scheduling selfish related machines \cite{ArcTar01}. In this problem, we are given several indivisible items (jobs) each of them with some size. Each item must be assigned to some agent (machine) and the goal is to minimize the maximum cost (the makespan). Note that the allocation of each machine is the sum of the size of the jobs allocated to this machine.

\begin{example}\label{exa:scheduling}
	Consider three machines and three jobs of size $10$, $6$, and $6$. For $L=1$ and $H=2+\epsilon$, for some small $\epsilon$ to be specified (below). The allocation of the jobs minimizing the \emph{makespan} for types $\theta=(L,L,H)$ and $\hat{\theta}=(L,H,H)$, for any $0 < \epsilon <2/3$, is as follows: $A(L,L,H) = (6+6,10,0)$ and $A(L,H,H)=(10,6,6)$.
	This is unique up to a permutation of
	the allocation of machines with the same type.
\end{example}
\newcommand{\valepsilon}{\frac{4}{\sqrt{3}}-2}
\newcommand{\valrho}{\frac{2}{\sqrt{3}}}

Using this example  we can show that our construction cannot lead to bribeproof mechanisms for minimizing the makespan in the 
scheduling problem above, or even to approximate the makespan within some small factor $\alpha>1$, i.e., returning an allocation whose makespan is at most $\alpha$ times the optimum makespan.

\begin{theorem}[selfish related machines]\label{th:makespan:no-pmon}
No bribeproof \linearmechanism{\lambda} for the  makespan minimization on three agents with two values-domains can approximate the makespan within a factor smaller than $\valrho \approx 1.1547$.
\end{theorem}

We note that the same impossibility result applies to randomized mechanisms using algorithms that 
pick an optimal allocation with some probability distribution:

\begin{remark}[randomized allocations]
	It is in principle possible to consider randomized allocations in which $a_i$ is a random variable and the allocation is 
	given by a probability distribution over those minimizing the makespan. In particular, we could define an algorithm which permutes the jobs allocated to machines of the same type. For the instance used to prove Theorem~\ref{th:makespan:no-pmon}, when the types are $(L,L,H)$ there are two optimal allocations,
	\begin{align*}
		(12,10,0)& & \mbox{and}& & (10,12,0)
	\end{align*}
	and by picking one of them with uniform probability the resulting algorithm $A$
	returns 
	\begin{align*}
		A(L,L,H) &= (11,11,0) & \mbox{and}& & A(L,H,H) &= (10,6,6).
	\end{align*}
	The same argument used to prove Theorem~\ref{th:makespan:no-pmon} applies also in this case and thus no 
	\linearmechanism{\lambda} can be bribeproof in expectation.
\end{remark}

\begin{remark}
	We stress that the previous impossibility results are conditioned to our \emph{\linearmechanism{\lambda}} for every choice of $\lambda$. Whether there exist  bribeproof mechanisms for makespan minimization at all is an interesting open question. Note that strategyproof mechanisms minimizing the makespan do exist \citep{ArcTar01} as well as computationally-efficient ones which guarantee $(1+\epsilon)$-approximation in polynomial time  for every fixed $\epsilon>0$ \citep{ChrKov08}. 
\end{remark}

\section{Collusion-proofness}\label{sec:collusion-proof}

In this section we present an application of our general construction: A sufficient condition for  which the \linearmechanism{1} is \emph{collusion-proof}, while \linearmechanism{\left(\frac{1}{2}\right)} is not even bribeproof. 

\begin{definition}[collusion-proof]
	A mechanism is collusion-proof if no coalition $C$ can improve its total utility by a join misreport of the types.  That is, 
	for every $\theta$ and for every coalition $C \subseteq \{1,\ldots,n\}$ 
	\begin{equation} \sum_{i\in C} u_i(\theta;\theta_i) \geq  \sum_{i\in C} u_i(\hat\theta;\theta_i) \label{eq:collusion-proof}\end{equation}
	where $\hat \theta$ is any type vector which differ from $\theta$ only in (some of) the types of agents in $C$, and $u_i(\hat\theta;\theta_i) = P_i(\hat \theta) - A_i(\hat \theta)\cdot \theta_i$ is the utility of agent $i$ when reporting $\hat \theta$; $u_i(\theta;\theta)$ is defined analogously. 
\end{definition}

We consider \emph{non-binary} allocations and prove that the following two conditions are sufficient for collusion-proofness: 
\begin{description}
	\item[Pareto efficiency:] If there is at least one agent of type $L$, then all agents of type $H$ should get zero allocation,
	\[ A_i(\theta)>0 \Rightarrow \theta_i=L \mbox{ or } \theta = \mathbf{H}, \]
	where $\mathbf{H}:=(H,\ldots,H)$ denotes the vector in which all types are $H$. 
	\item[Minimality:] For types $\mathbf{H}$, no group of agents can decrease its total allocation by changing some of its types, 
	\[ \sum_{i\in C} A_i(\mathbf{H}) \leq \sum_{i\in C} A_i(\hat\theta_C,\mathbf{H}_{-C}), \]
	for any vector $(\hat\theta_C,\mathbf{H}_{-C})$ in which all types not in $C$ are $H$.  
\end{description}

\begin{theorem}
	The \linearmechanism{1} is collusion-proof if algorithm $A$ satisfies Pareto efficiency and minimality.
\end{theorem}
\begin{proof}
	Observe that in the \linearmechanism{1} the utility of each agent is zero or negative, and a negative utility occurs only when the agent under consideration has type $H$ and she gets a positive allocation. 
	We show that \eqref{eq:collusion-proof} holds  for every $\theta$, for every coalition $C \subseteq \{1,\ldots,n\}$, and for every $\hat \theta$ which differ from $\theta$ only in (some of) the types of agents in $C$.
	
	For $\theta\neq \mathbf{H}$, every truth-telling agent has \emph{zero} utility  (because of Pareto efficiency), and thus \eqref{eq:collusion-proof} follows from the observation that utilities are nonnegative.
	
	For $\theta = \mathbf{H}$, condition \eqref{eq:collusion-proof} follows from minimality and by the observation that in \linearmechanism{1} the utility of agents of type $H$ is of the form 
	\[u_i(\hat\theta;H) = -(H-L)A_i(\hat \theta). \]
	This completes the proof.
\end{proof}

We next give two examples of such allocation rules. The first one is the \emph{uniform rule} which  divides a perfectly-divisible good equally among all agents having the smallest type:
\begin{equation}
U_i(\theta)=\begin{cases}
1/n & \mbox{if $\theta=(H,\ldots,H)$},\\
1/n_L & \mbox{if $\theta \neq (H,\ldots,H)$ and $\theta_i=L$},\\
0 & \mbox{if $\theta \neq (H,\ldots,H)$ and $\theta_i=H$.}
\end{cases}\label{exa:uniform-rule}
\end{equation}
(Here $n_L$ denotes the number of agents of type $L$ in $\theta$.)
\begin{corollary}
	The \linearmechanism{1} with uniform rule \eqref{exa:uniform-rule} is collusion-proof. 
\end{corollary}

The second rule, defined for the case of three agents, is a simple modification of the previous one so that  unequal amounts are allocated when the three types are identical:
\begin{align}
	A(\theta) = \begin{cases}
		(2/3,1/6,1/6) & \mbox{if $\theta\in\{(L,L,L),(H,H,H)\}$},\\
		U(\theta) & \mbox{otherwise.}
	\end{cases}\label{exa:uniform-rule-tilted}
\end{align}
To see that also this rule satisfies minimality, observe that $\sum_{i\in C} A_i(\hat\theta) = 1$
for all $\hat\theta=(\hat\theta_C,\mathbf{H}_{-C}) \neq \mathbf{H}$. The next simple corollary shows that the uniform rule is not the only collusion-proof rule. Moreover, it shows that in some applications $\lambda=1$ gives the correct choice for \linearmechanism{\lambda}.

\begin{corollary}\label{cor:bribeproof}
	For three agents, the \linearmechanism{1} with rule in \eqref{exa:uniform-rule-tilted} is collusion-proof, while the \linearmechanism{\left(\frac{1}{2} \right) } with the same rule is not even bribeproof.
\end{corollary}

\section{Conclusion and open questions}
This work provides a sufficient condition for obtaining  \emph{non-trivial} bribeproof mechanisms on certain \emph{restricted domains}. Specifically, the proposed  \emph{two-values domains} can be regarded as the simplest domains which still capture some ``combinatorial'' structure of the problem. For instance, it may be the case that two particular agents cannot be selected simultaneously, and therefore have a ``conflict of interests''. 
In our opinion, the simple construction proposed here is interesting for the following reasons:
\begin{itemize}
	\item The impossibility results known in the literature \citep{Sch00,Miz03} do not apply as the mechanisms are far from trivial (``fixed solution'') since they maximize the social welfare in binary allocation problems.
	\item In some natural welfare-maximization problems, the proposed mechanisms are essentially \emph{the only} bribeproof ones. The results are also the best possible in the sense that there is no mechanism on more general domains, nor a mechanism dealing with coalitions of size more than two. 
\end{itemize}
On the one hand, this suggests that the results are tight for welfare maximization problems. On the other hand,  it would be interesting to understand, or even characterize, bribeproofness in finite domains (without appealing to welfare-maximization). The impossibility results for scheduling problems presented in Section~\ref{sec:non-binary}  apply only to our construction.  Whether other bribeproof mechanisms exist for this problem is an interesting open question. More generally, the results suggest to investigate the trade-off between the social welfare and the ``richness'' of the domain. In that sense, our results can be seen as ``complementary'' of the impossibility results in \cite{Sch00,Miz03}: On general domains nothing is possible, while optimal social welfare for certain problems can be achieved in two-values domains but not on three-values domains. Therefore, one might consider \emph{approximate} social welfare in discrete domains with ``few'' possible values (e.g., approximate path-auction with three or more values). More generally, the trade-off may also involve the type of incentive compatibility condition that we want. For example, \cite{DutGkaRou14} showed that a certain class of weakly group strategyproof mechanisms (studied in \cite{MilSeg14}) can only achieve an approximate social welfare in binary one-parameter domains.
	
Concerning the construction, we note that \linearmechanism{\lambda}s guarantee bribeproofness (or even stronger conditions) for different classes of problems. The choice of the parameter $\lambda$  plays a central role (this specifies the fixed payment per unit of work provided to the agents) and, in particular,  different problems require a different $\lambda$ (cf., Corollary~\ref{cor:sp:characterization}, Theorem~\ref{th:parallel:three-vals:char}, and Corollary~\ref{cor:bribeproof}).

\paragraph{Acknowledgments.} We are grateful to the anonymous reviewers for a through reading of the paper, and for suggesting a possible connection with the work \cite{ohseto2000strategy}.

\bibliographystyle{plain}
\bibliography{bribeproof_short}

\newpage
\appendix

 \section{Restricted combinatorial auction}\label{sec:restricted CA}
 Consider the (restricted) combinatorial auction with \emph{known single-minded bidders} \citep{MuaNis08}: each agent $i$ is bidding only for one subset of items $S_i$ and this subset is \emph{publicly known}, while her valuation $v_i$ for this subset is private (the type of the agent). If we further restrict the valuations to two values, 
 \[
 v_i \in \{v_{low}, v_{high}\} 
 \]
 then we have a problem with binary allocations over a two-values domain: It is enough to consider $a_i$ being $1$ if $i$ gets all items in her desired subset and $0$ otherwise, and the two-values domain with ``possible costs'' 
 \begin{align}
 	L &=-v_{high} & \mbox{and}& & H&=-v_{low}.
 \end{align}
 We can thus apply Theorem~\ref{th:binary-utilitarian} and obtain the following bribeproof mechanism maximizing the social welfare:
 \begin{example}[restricted combinatorial auction]
 	The following  \linearmechanism{\left(\frac{1}{2}\right)} is a is a bribeproof mechanism maximizing the social welfare in 
 	combinatorial auctions with known single-minded bidders with two-values domains (valuations of the subset of preferred items is either $v_{low}$  or $v_{high}$). The
 	algorithm sets a strict ordering of all binary allocations (resulting from the assignment of items to the agents). 
 	Given the agents bids, the algorithm computes the allocation $a^*$ maximizing the social welfare, breaking ties between solutions using the fixed strict
 	ordering. Each agent getting her preferred bundle (i.e., $a^*_i=1$) pays an amount equal to  $\frac{v_{low}+v_{high}}{2}$.
 \end{example}
 
 We next show that, in some instances, this is \emph{the only} bribeproof mechanism over two-values domains maximizing the social welfare. 
 In particular, consider the case of four agents whose subsets of preferred items intersect as follows:

 
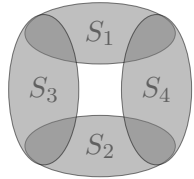
\begin{center}
 	\begin{tikzpicture}
 	\node[ellipse,draw, fill=gray, opacity=.5,minimum height=0.5cm,minimum width=2cm] (s1) at (1.75,1.75) {$S_1$};
 	\node[ellipse,draw, fill=gray, opacity=.5,minimum height=0.5cm,minimum width=2cm] (s2) at (1.75,0.25) {$S_2$};
 	\node[ellipse,draw, fill=gray, opacity=.5,minimum height=2cm,minimum width=0.5cm] (s3) at (1,1) {$S_3$};
 	\node[ellipse,draw, fill=gray, opacity=.5,minimum height=2cm,minimum width=0.5cm] (s3) at (2.5,1) {$S_4$};
 	\end{tikzpicture}
 \end{center}
 Suppose further that each item is available in one copy, thus implying that 
 if two subsets intersect only one of the two agents can get her preferred subset.  Then the following is a simple observation:
 \begin{obs}
 	On instances as above, 
 	if a mechanism maximizes the social welfare, then it must allocate the items so that either both agents 
 	$1$ and $2$ get their preferred items or agents $3$ and $4$ get their preferred items.
 \end{obs}
 This observation allows us to show that these instances are equivalent to the path auction instances used in Theorem~\ref{th:sp:generalized-domain} (see Figure~\ref{fig:sp-diamond}). Indeed, the previous observation says that the mechanism must return one of the following two allocations: 
 \begin{align}
 	&(1,1,0,0) & \mbox{or} & & (0,0,1,1)
 \end{align}
 The second key observation is that the social welfare of these two allocations are related to the cost in the path auction:
 \begin{obs}
 	Allocation $(1,1,0,0)$ for the restricted combinatorial auction has social welfare $v_1 + v_2 = -(\theta_1 + \theta_2)$, while 
 	allocation $(0,0,1,1)$ has social welfare $v_3 + v_4 = -(\theta_3 + \theta_4)$
 \end{obs}

 We thus have that the minimal-cost path in the graph in Figure~\ref{fig:sp-diamond} corresponds to the maximal-welfare
 allocation in the combinatorial auction above and vice versa.
 The characterization for path auction mechanisms translate into the following result:

 \begin{corollary}\label{cor:RCA:char}
 	The \linearmechanism{\left(\frac{1}{2}\right)} is the only bribeproof mechanism maximizing the 
 	social welfare in some instances of combinatorial auctions with known single-minded bidders over two-values domains. 
 \end{corollary}
 
 From this we derive immediately certain inherent limitations of bribeproof mechanisms: Either the mechanism charges 
 some agent a price higher than her valuation or, by rescaling the payments, the mechanism is actually paying some agents and the overall revenue is \emph{negative}.
 
 Finally we observe that the equivalence between certain instances of restricted combinatorial auctions (with known single-minded bidders) and instances of path auctions established above allows us to obtain an impossibility result directly from Theorem~\ref{th:sp:no-three-vals}.
 
 \begin{corollary}
 	There is no bribeproof mechanism maximizing the 
 	social welfare in some instances of combinatorial auctions with known single-minded bidders over three-values domains. 
 \end{corollary}

\section{Bribeproofness versus strategyproofness}
We show that, even in our two-values domains, bribeproofness is provably more restricting that strategyproofness. In particular, there exist algorithms which admit strategyproof mechanisms 
but no bribeproof ones.

\begin{lemma}\label{le:bossy:paychange}
	If a mechanism $(A,p)$ is bribeproof for binary allocation problems, then the following implication must hold for all $i$ and for all $\theta$:
	\begin{align*}
		\influence{i} &= 0 & \mbox{$\Rightarrow$}& & \partialpay{j} &= \theta_j \cdot \influence{j}.
	\end{align*}
\end{lemma}
\begin{proof}
	The following implication is due to strategyproofness for agent $i$: 
	\[
	\influence{i} = 0 \ \ \ \mbox{ $\Rightarrow$ } \ \ \ \partialpay{i}=0. 
	\]
	Thus the utility of $i$ is constant (as a function of $i$'s reported type) and bribeproofness is equivalent to the condition that also the utility of $j$ is constant (as a function of $i$'s reported type), 
	\[
	p_j(L,\theta_{-i}) - \theta_jA_j(L,\theta_{-i}) = p_j(H,\theta_{-i}) - \theta_jA_j(H,\theta_{-i})
	\]
	since otherwise $j$ can bribe $i$ to report the type that gives the highest utility to $j$. This equality is 
	the same as $\partialpay{j} = \theta_j \cdot \influence{j}$.
\end{proof}

\begin{example}[non-bribeproof rule]
	The following algorithm $A$ is clearly monotone \eqref{eq:mon} and thus it admits a strategyproof mechanism: 
	\[
	A(L,L)=(1,0), \ \  A(H,L)=A(H,H)=(0,1), \ \ A(L,H)=(0,0).
	\]
	We next show that no mechanism $(A,p)$ can be bribeproof.  
	To see this observe that we can rescale the payments so that 
	\[
	p(L,H)=(0,0).
	\]
	By Fact~\ref{fac:adj:pay} we have 
	\[
	p(H,L)=p(H,H).
	\]
	Now by applying Lemma~\ref{le:bossy:paychange} we get
	\begin{align*}
		p(L,L) &= (L,0) & \mbox{and}& & p(H,L)&=p(H,H) = (0,H).
	\end{align*}
	This violates bribeproofness for $\theta=(L,L)$ and $\hat \theta=(H,L)$. 
\end{example}

\section{Postponed proofs}
In some of the proofs below we shall use the following notation for the agents' utilities. 
The utility of agent $i$ for a mechanism $(A,p)$ under consideration is denoted as
\[
u_i(\hat \theta;\theta_i) := p_i(\hat \theta) - A_i(\hat \theta)\cdot \theta_i.
\]

\subsection{Strong bribeproofness (and Proof of Theorem~\ref{th:sbribeproof:equivalence})}
To obtain \sbribeproof\ mechanisms we need and additional condition on the algorithm:

\begin{lemma}\label{le:sbribeproof:necessary}
	If a mechanism is \sbribeproof\ for two-values domains, then its algorithm $A$ must satisfy, for all $\theta\in \Theta$ and for all $i,j\in N$, the following conditions:
	\begin{align}
		A_i(L,L,\theta_{-ij})+A_j(L,L,\theta_{-ij})&\geq A_i(H,H,\theta_{-ij})+A_j(H,H,\theta_{-ij}), & & \label{eq:sbribeproof:sufficient:LL-HH}\\ 
		A_i(L,H,\theta_{-ij})-A_j(L,H,\theta_{-ij})&\geq A_i(H,L,\theta_{-ij})-A_j(H,L,\theta_{-ij}). \label{eq:sbribeproof:sufficient:LH-HL}
	\end{align}
\end{lemma}

\begin{proof}
	The proof of this lemma consists in considering the two agents of a possible coalition  as a single agent with a ``multi-dimensional'' type, and then imposing strategyproofness to this agent.
	  
Strong bribeproofness requires \eqref{eq:bribeproof:breibee} to hold for $\theta$ and $\hat \theta=(\hat \theta_i,\hat \theta_j,\theta_{-ij})$.
For ease of notation, we let 
\begin{align*}
 \bar{p}(\theta_i,\theta_j):=&p(\theta_i,\theta_j,\theta_{-ij}) & \mbox{and}& & \bar A(\theta_i,\theta_j):=&A(\theta_i,\theta_j,\theta_{-ij}).
\end{align*}
By writing down condition \eqref{eq:bribeproof:breibee} explicitly we get the following inequalities:
 \begin{align*}
 \intertext{For $\theta=(L,L,\theta_{-ij})$ and $\hat \theta=(H,H,\theta_{-ij})$}
 & \bar{p}_i(L,L) - \bar{A}_i(L,L)L & &+ & \bar{p}_j(L,L)- \bar{A}_j(L,L)L & &\geq \\
 & \bar{p}_i(H,H) - \bar{A}_i(H,H)L & &+ & \bar{p}_j(H,H)- \bar{A}_j(H,H)L.& \\
 \intertext{For $\theta=(H,H,\theta_{-ij})$ and $\hat \theta=(L,L,\theta_{-ij})$}
 & \bar{p}_i(H,H) - \bar{A}_i(H,H)H & &+ & \bar{p}_j(H,H)- \bar{A}_j(H,H)H & &\geq \\
 & \bar{p}_i(L,L) - \bar{A}_i(L,L)H & &+ & \bar{p}_j(L,L)- \bar{A}_j(L,L)H.& 
 \intertext{By summing these two inequalities we obtain \eqref{eq:sbribeproof:sufficient:LL-HH}. The proof of \eqref{eq:sbribeproof:sufficient:LH-HL} is similar: }
  \intertext{For $\theta=(L,H,\theta_{-ij})$ and $\hat \theta=(H,L,\theta_{-ij})$}
 & \bar{p}_i(L,H) - \bar{A}_i(L,H)L & &+ & \bar{p}_j(L,H)- \bar{A}_j(L,H)H & &\geq \\
 & \bar{p}_i(H,L) - \bar{A}_i(H,L)L & &+ & \bar{p}_j(H,L)- \bar{A}_j(H,L)H.&
 \intertext{For $\theta=(H,L,\theta_{-ij})$ and $\hat \theta=(L,H,\theta_{-ij})$}
 & \bar{p}_i(H,L) - \bar{A}_i(H,L)H & &+ & \bar{p}_j(H,L)- \bar{A}_j(H,L)L & &\geq \\
 & \bar{p}_i(L,H) - \bar{A}_i(L,H)H & &+ & \bar{p}_j(L,H)- \bar{A}_j(L,H)L.&
 \end{align*}
and by summing these two inequalities we obtain \eqref{eq:sbribeproof:sufficient:LH-HL}.
\end{proof}

\begin{theorem}\label{th:sbribeproof:linear}
	The \linearmechanism{\left(\frac{1}{2}\right)} is \sbribeproof\ for two-values domains if and only if its algorithm $A$ satisfies bounded influence and the two conditions in Lemma~\ref{le:sbribeproof:necessary}.
\end{theorem}

\begin{proof}
	The necessary part follows from  Corollary~\ref{cor:wokload:bribeproof:1/2} and by Lemma~\ref{le:sbribeproof:necessary}. We thus only show that these conditions are sufficient for strong bribeproofness. By Corollary~\ref{cor:wokload:bribeproof:1/2}, we already have that the mechanism is bribeproof. Therefore it only remains to prove that  \eqref{eq:bribeproof:breibee} holds for any two 
	$\theta$ and $\hat \theta=(\hat \theta_i,\hat \theta_j,\theta_{-ij})$ which differ in both the types of $i$ and $j$, that is,  $\hat \theta_i\neq \theta_i$ and $\hat \theta_j \neq \theta_j$. The proof is analogous to the proof of Theorem~\ref{th:wokload:bribeproof}. In particular, we use the fact that utilities are of the form \eqref{eq:linear-mechanism:utility} and consider the following four possible cases:
	\begin{description}
		\item[($\theta=(L,L,\theta_{-ij})$ and $\hat \theta=(H,H,\theta_{-ij})$)] This case   corresponds to
		\[
		\frac{1}{2}(A_i(L,L,\theta_{-ij})+A_j(L,L,\theta_{-ij})) \geq \frac{1}{2}(A_i(H,H,\theta_{-ij})+A_j(H,H,\theta_{-ij}))
		\]
		which is equivalent to \eqref{eq:sbribeproof:sufficient:LL-HH}.
		
		\item[($\theta=(H,H,\theta_{-ij})$ and $\hat \theta=(L,L,\theta_{-ij})$)] This case   corresponds to
		\[
		-\frac{1}{2}(A_i(H,H,\theta_{-ij})+A_j(H,H,\theta_{-ij})) \geq -\frac{1}{2}(A_i(L,L,\theta_{-ij})+A_j(L,L,\theta_{-ij}))
		\]
		which is equivalent to \eqref{eq:sbribeproof:sufficient:LL-HH}.
		
		\item[($\theta=(L,H)$ and $\hat \theta=(H,L)$)] This case   corresponds to
		\[
		\frac{1}{2}(A_i(L,H,\theta_{-ij})-A_j(L,H,\theta_{-ij})) \geq \frac{1}{2}(A_i(H,L,\theta_{-ij})-A_j(H,L,\theta_{-ij}))
		\]
		which is equivalent to \eqref{eq:sbribeproof:sufficient:LH-HL}.
		\item[($\theta=(H,L)$ and $\hat \theta=(L,H)$)] This case   corresponds to
		\[
		\frac{1}{2}(-A_i(H,L,\theta_{-ij})+A_j(H,L,\theta_{-ij})) \geq \frac{1}{2}(-A_i(L,H,\theta_{-ij})+A_j(L,H,\theta_{-ij}))
		\]
		which is equivalent to \eqref{eq:sbribeproof:sufficient:LH-HL}.
	\end{description}
	This completes the proof.
\end{proof}

\subsubsection{Proof of Theorem~\ref{th:sbribeproof:equivalence}}
\begin{proof}
	Since for binary allocations  the work allocated to each agent is either $0$ or $1$,
	non-bossiness and monotonicity \eqref{eq:mon} are equivalent to bounded influence \eqref{eq:bounded-influence}. 
	This gives the equivalence between the statements in Item~\ref{th:sbribeproof:equivalence:bribeproof} and Item~\ref{th:sbribeproof:equivalence:non-bossiness}.  
	
	Since \sbribeproof ness implies bribeproofness, in order to complete the proof it suffices to show that monotonicity and non-bossiness imply strong bribeproofness. To get to this conclusion, we show that $A$ satisfies also conditions \eqref{eq:sbribeproof:sufficient:LL-HH}-\eqref{eq:sbribeproof:sufficient:LH-HL} in Lemma~\ref{le:sbribeproof:necessary}, as required by Theorem~\ref{th:sbribeproof:linear}. 
	For the sake of readability, for every $\theta_{-ij}$, we consider the resulting algorithm $A(\cdot,\cdot,\theta_{-ij})$ as a two-agents algorithm $B(\cdot,\cdot)$ where the first and the second agents correspond to $i$ and $j$, respectively:
	\[
	B(x,y):= \left(A_i(x,y,\theta_{-ij}), A_j(x,y,\theta_{-ij}) \right).
	\]
	We have the following implications:
	\begin{align*}
		B(L,L)&=(0,0) & \Rightarrow& &  B(H,L)&=(0,0) & \Rightarrow& &  B(H,H)&=(0,0) \\
		B(L,L)&=(0,1) & \Rightarrow& &  B(H,L)&=(0,1) & \Rightarrow& &  B(H,H)&\neq(1,1)
		\intertext{where each implication uses both monotonicity and non-bossiness. These implications yield \eqref{eq:sbribeproof:sufficient:LL-HH}.
			In a similar way we obtain the following implication:}
		B(L,H)&=(0,1) & \Rightarrow& &  B(H,H)&=(0,1) & \Rightarrow& &  B(H,H)&=(0,1) 
	\end{align*}
	which yields \eqref{eq:sbribeproof:sufficient:LH-HL}. 
	We have thus shown that the statement of Item~\ref{th:sbribeproof:equivalence:non-bossiness}  implies that the conditions of Lemma~\ref{le:sbribeproof:necessary} hold. By Theorem~\ref{th:sbribeproof:linear}, we have that the statement of Item~\ref{th:sbribeproof:equivalence:non-bossiness} implies that of Item~\ref{th:sbribeproof:equivalence:sbribeproof}. Since the latter implies the statement of Item~\ref{th:sbribeproof:equivalence:bribeproof}, the theorem follows.
\end{proof}

\subsection{Proof of Theorem~\ref{th:sp:generalized-domain}}\label{sec:th:sp:generalized-domain:proof}
We have to prove that, for any $\epsilon\geq 0$, if $(A,p)$ is a bribeproof mechanism for 
the path auction with $\epsilon$-perturbed domain, then $(A,p)$ is a \linearmechanism{\left(\frac{1}{2}\right)}. That is, there exist constants $\{f_i\}_{i\in N}$ such that, for any $\theta$ in such domain, the payments satisfy
\begin{equation}
\label{eq:th:sp:generalized-domain:linear-payments}
p_i(\theta) = f_i + 
\begin{cases}
\frac{L+H}{2} & \mbox{if $i$ is selected for types $\theta$,} \\
\phantom{-}0  & \mbox{otherwise.}
\end{cases}
\end{equation}

We first prove that the payments of the mechanism depend uniquely on the selected path.  

\begin{nclaim}\label{cl:sp:proof:payments-solution-based}
	There exist $p^{(upper)}$ and $p^{(lower)}$ such that, for any $\theta$, it holds that
	\[
	p(\theta) =
	\begin{cases}
	p^{(upper)} & \mbox{if $1$ and $2$ are selected for types $\theta$,} \\
	p^{(lower)} & \mbox{if $3$ and $4$ are selected for types $\theta$.}
	\end{cases}
	\]   
\end{nclaim}

\begin{proof}
	Consider the following two instances having a unique shortest path, 
	\begin{align*}
		\stpathup{L-\epsilon}{L-\epsilon}{H}{H}& &
		\stpathdown{H+\epsilon}{H+\epsilon}{L}{L}&
	\end{align*}
	and let 
	$p^{(upper)} := p(L-\epsilon,L-\epsilon,H,H)$ and $p^{(lower)} := p(H+\epsilon,H+\epsilon,L,L)$ be the corresponding payments. 
	
	If $\theta$ is a type vector for which the mechanism selects the upper path, then the same must also happen if we increase the cost of edges in the lower path and decrease that of edges in the upper path: 
	\begin{align*}
		A(\theta)&=A(L-\epsilon,\theta_2,\theta_3,\theta_4)=A(L-\epsilon,L-\epsilon,\theta_3,\theta_4)  \\
		&= A(L-\epsilon,L-\epsilon,H,\theta_4)=A(L-\epsilon,L-\epsilon,H,H).
	\end{align*}
	Fact~\ref{fac:adj:pay} implies that  $p(\theta)=p(L-\epsilon,L-\epsilon,H,H)=p^{(upper)}$.
	
	Similarly, if $\theta$ is a type vector for which the mechanism selects the lower path, then 
	\begin{align*}
		A(\theta)&=A(H+\epsilon,\theta_2,\theta_3,\theta_4)=A(H+\epsilon,H+\epsilon,\theta_3,\theta_4)  \\
		&= A(H+\epsilon,H+\epsilon,L,\theta_4)=A(H+\epsilon,H+\epsilon,L,L).
	\end{align*}
	Fact~\ref{fac:adj:pay} implies that  $p(\theta)=p(H+\epsilon,H+\epsilon,L,L)=p^{(lower)}$.
\end{proof}

From the previous claim we obtain that every agent receives a fixed compensation $f_i$ plus and an additional compensation $q_i$ whenever selected.

\begin{nclaim}\label{cl:sp:proof:payments-fixed-plus-served}
	There exist constants $f_i$ and $q_i$ such that, for any $\theta$, it holds that 
	\[
	p_i(\theta) = f_i + 
	\begin{cases}
	q_i & \mbox{if $i$ is selected for types $\theta$,} \\
	0  & \mbox{otherwise.}
	\end{cases}
	\]                                                                                                                                                                    
\end{nclaim}

\begin{proof}
	Claim~\ref{cl:sp:proof:payments-fixed-plus-served} follows directly from Claim~\ref{cl:sp:proof:payments-solution-based}  by setting $f_i$ and $q_i$ as follows:
	\begin{align*}
		\mbox{ for }i&\in \{1,2\}, & f_i&:= p_i^{(lower)} & \mbox{and}& & q_i&:= p_i^{(upper)} - p_i^{(lower)}; \\ 
		\mbox{ for }i&\in \{3,4\}, & f_i&:= p_i^{(upper)} & \mbox{and}& & q_i&:= p_i^{(lower)} - p_i^{(upper)}.  
	\end{align*}
	
\end{proof}

In order to conclude that the mechanism must be a \linearmechanism{\left(\frac{1}{2}\right)} it is enough to prove that  $q_i=\frac{L+H}{2}$ for all $i$. To this end, we first show the following intermediate result.

\begin{nclaim}\label{cla:sp:characterization:sum-payments:heterogeneous}The constants $q_i$ in Claim~\ref{cl:sp:proof:payments-fixed-plus-served} must satisfy
	\begin{equation}\label{eq:sp:characterization:sum-payments:heterogeneous}
	q_1 + q_2 = L+H = q_3 + q_4.
	\end{equation}
\end{nclaim}
\begin{proof}

	Let us first consider the case the mechanism breaks ties as follows in this particular instance:
	\begin{equation}
	\label{eq:sp:heterogeneous:tie-breaking-one}
	\stpathup{L-\epsilon}{H+\epsilon}{L}{H}
	\end{equation}
	From this we obtain the following four inequalities which imply \eqref{eq:sp:characterization:sum-payments:heterogeneous}:
	\begin{description}                         
		\item[($q_1+q_2 \geq L+H$).] By contradiction, if $q_1+q_2 < L+H$ then the coalition given by the two agents in the upper path can improve in this scenario:
		\begin{align*}
			\stpathup{L-\epsilon}{H+\epsilon}{L}{H}& & \leftarrow \theta& & \hat \theta&\rightarrow & \stpathdown{H+\epsilon}{H+\epsilon}{L}{H}&
		\end{align*}
		(Here and in the rest of this proof, the left picture shows the true types $\theta$ and the right one the types $\hat \theta$ for which bribeproofness \eqref{eq:bribeproof:breibee} is violated.)
		\item[($q_3+q_4 \leq L+H$).] By contradiction, if $q_3+q_4 > L+H$ then the coalition given by the two agents in the lower path can improve in this scenario:
		\begin{align*}
			\stpathup{L-\epsilon}{H+\epsilon}{L}{H}& & 
			\leftarrow \theta& & \hat \theta&\rightarrow &
			\stpathdown{L-\epsilon}{H+\epsilon}{L}{L}&
		\end{align*}
		\item[($q_1+q_2 \leq L+H$).] By contradiction, if $q_1+q_2 > L+H$ then  the coalition given by the two agents in the upper path violates bribeproofness \eqref{eq:bribeproof:breibee} for $\theta$ and $\hat \theta$ being the types on the left and on the right of this picture:
		\begin{align*}
			\stpathdown{L-\epsilon}{H+\epsilon}{L}{L}& & 
			\leftarrow \theta& & \hat \theta&\rightarrow &
			\stpathup{L-\epsilon}{L-\epsilon}{L}{L}&
		\end{align*}
		Note that for $\epsilon>0$ the upper path must be selected for the right instance. We next show that this must also happen for $\epsilon =0$ since otherwise, because  $q_1+q_2 > L+H \geq q_3+q_4$,  
		\begin{align*}
			\stpathdown{L}{L}{L}{L}& &
			\leftarrow \theta& & \hat \theta&\rightarrow &
			\stpathup{L}{L}{H}{L}&
		\end{align*}
		for the case agent $3$ satisfies $q_3<\min\{q_1,q_2\}$. The case 
		$q_4<\min\{q_1,q_2\}$ is analogous.
		\item[($q_3+q_4 \geq L+H$).] By contradiction, assume $q_3+q_4 < L+H$. Let us first observe that for every instance in which both paths cost $L+H$ the mechanism must break ties in favor of the upper path like in \eqref{eq:sp:heterogeneous:tie-breaking-one}.
		
		By the inequality $q_1 + q_2 \geq L+H$ proven above, there must be one agent $k\in \{1,2\}$ with $q_k\geq \frac{L+H}{2}$. 
		From $q_3+q_4 < L+H$ it also follows that there must be one agent $l\in \{3,4\}$ with $q_l < \frac{L+H}{2}$. Thus, there are four possible cases to consider:
		\begin{description}
			\item[($q_1\geq \frac{L+H}{2}$ and $q_3< \frac{L+H}{2}$).] The coalition given by  agents $1$ and $3$ can improve in this scenario:
			\begin{align*}
				\stpathdown{L-\epsilon}{H+\epsilon}{L}{L}& &
				\leftarrow \theta& & \hat \theta&\rightarrow &
				\stpathup{L-\epsilon}{H+\epsilon}{H}{L}& 
			\end{align*}
			where the upper path is selected on the right instance because of the initial observation.
			\item[($q_1\geq \frac{L+H}{2}$ and $q_4< \frac{L+H}{2}$).] The coalition given by  agents $1$ and $4$ can improve in this scenario:
			\begin{align*}
				\stpathdown{L-\epsilon}{H+\epsilon}{L}{L}& &
				\leftarrow \theta& & \hat \theta&\rightarrow &
				\stpathup{L-\epsilon}{H+\epsilon}{L}{H}& 
			\end{align*}
			where the upper path is selected on the right instance because of the initial observation.
			\item[($q_2\geq \frac{L+H}{2}$ and $q_3< \frac{L+H}{2}$).] The coalition given by  agents $2$ and $3$ can improve in this scenario:
			\begin{align*}
				\stpathdown{H+\epsilon}{L-\epsilon}{L}{L}& &
				\leftarrow \theta& & \hat \theta&\rightarrow &
				\stpathup{H+\epsilon}{L-\epsilon}{H}{L}& 
			\end{align*}
			where the upper path is selected on the right instance because of the initial observation.
			\item[($q_2\geq \frac{L+H}{2}$ and $q_4< \frac{L+H}{2}$).] The coalition given by  agents $2$ and $4$ can improve in this scenario:
			\begin{align*}
				\stpathdown{H+\epsilon}{L-\epsilon}{L}{L}& & 
				\leftarrow \theta& & \hat \theta&\rightarrow &
				\stpathup{H+\epsilon}{L-\epsilon}{L}{H}& 
			\end{align*}
			where the upper path is selected on the right instance because of the initial observation.
		\end{description}
	\end{description}
	We have thus shown that \eqref{eq:sp:characterization:sum-payments:heterogeneous} holds if the mechanism breaks ties as in \eqref{eq:sp:heterogeneous:tie-breaking-one}.
	
	Let us now consider the case the mechanism does not break ties as in \eqref{eq:sp:heterogeneous:tie-breaking-one}, that is, 
	the mechanism breaks ties as follows in that instance:
	\begin{equation}
	\label{eq:sp:heterogeneous:tie-breaking-two}
	\stpathdown{L-\epsilon}{H+\epsilon}{L}{H}
	\end{equation}
	From this we obtain the following four inequalities which imply \eqref{eq:sp:characterization:sum-payments:heterogeneous}:
	\begin{description}                         
		\item[($q_3+q_4 \geq L+H$).] By contradiction, if $q_3+q_4 < L+H$ then the coalition given by the two agents in the lower path can improve in this scenario:
		\begin{align*}
			\stpathdown{L-\epsilon}{H+\epsilon}{L}{H}& & 
			\leftarrow \theta& & \hat \theta&\rightarrow &
			\stpathup{L-\epsilon}{H+\epsilon}{H}{H} & 
		\end{align*}
		\item[($q_1+q_2 \leq L+H$).] By contradiction, if $q_1+q_2 > L+H$ then  the coalition given by the two agents in the upper path can improve in this scenario:
		\begin{align*}
			\stpathdown{L-\epsilon}{H+\epsilon}{L}{H}& & 
			\leftarrow \theta& & \hat \theta&\rightarrow &
			\stpathup{L-\epsilon}{L-\epsilon}{L}{H}& 
		\end{align*}
		\item[($q_1+q_2 \geq L+H$).] By contradiction, if $q_1+q_2 < L+H$ then the coalition given by the two agents in the upper path can improve in this scenario:
		\begin{align*}
			\stpathup{L-\epsilon}{H+\epsilon}{H}{H}& &
			\leftarrow \theta& & \hat \theta&\rightarrow &
			\stpathdown{H+\epsilon}{H+\epsilon}{H}{H}& 
		\end{align*}
		Note that for $\epsilon>0$ the lower path must be selected for the right instance. We next show that this must also happen for $\epsilon =0$ since otherwise, because  $q_1+q_2 < L+H \leq q_3+q_4$,  
		\begin{align*}
			\stpathup{H}{H}{H}{H}& &
			\leftarrow \theta& & \hat \theta&\rightarrow &
			\stpathdown{H}{H}{L}{H}& 
		\end{align*}
		for the case agent $3$ satisfies $q_3>\max\{q_1,q_2\}$. The case 
		$q_4>\max\{q_1,q_2\}$ is analogous.
		
		\item[($q_3+q_4 \leq L+H$).] By contradiction, assume $q_3+q_4 > L+H$. Let us first observe that for every instance in which both paths cost $L+H$ the mechanism must break ties in favor of the lower path like in \eqref{eq:sp:heterogeneous:tie-breaking-two}. 
		
		By the inequality $q_1 + q_2 \geq L+H$ proven above, there must be one agent $k\in \{1,2\}$ with $q_k\geq \frac{L+H}{2}$. 
		From $q_3+q_4 < L+H$ it also follows that there must be one agent $l\in \{3,4\}$ with $q_l< \frac{L+H}{2}$. Thus, there are four possible cases to consider:
		\begin{description}
			\item[($q_3< \frac{L+H}{2}$ and $q_1 \geq \frac{L+H}{2}$).] The coalition given by  agents $1$ and $3$ can improve in this scenario:
			\begin{align*}
				\stpathdown{L-\epsilon}{H+\epsilon}{L}{H}& &
				\leftarrow \theta& & \hat \theta&\rightarrow &
				\stpathup{L-\epsilon}{H+\epsilon}{H}{H}&  
			\end{align*} 
			\item[($q_3 < \frac{L+H}{2}$ and $q_2 \geq  \frac{L+H}{2}$).] The coalition given by  agents $2$ and $3$ can improve in this scenario:
			\begin{align*}
				\stpathdown{H+\epsilon}{L-\epsilon}{L}{H}& & 
				\leftarrow \theta& & \hat \theta&\rightarrow &
				\stpathup{H+\epsilon}{L-\epsilon}{H}{H}& &  
			\end{align*}
			where the lower path is selected on the left instance because of the initial observation.
			\item[($q_4<\frac{L+H}{2}$ and $q_1 \geq \frac{L+H}{2}$).] The coalition given by  agents $1$ and $4$ can improve in this scenario:
			\begin{align*}
				\stpathdown{L-\epsilon}{H+\epsilon}{H}{L}& & 
				\leftarrow \theta& & \hat \theta&\rightarrow &
				\stpathup{L-\epsilon}{H+\epsilon}{H}{H}&  
			\end{align*}
			where the lower path is selected on the left instance because of the initial observation.
			\item[($q_4 < \frac{L+H}{2}$ and $q_2 \geq \frac{L+H}{2}$).] The coalition given by  agents $2$ and $4$ can improve in this scenario:
			\begin{align*}
				\stpathdown{H+\epsilon}{L-\epsilon}{H}{L}& &   
				\leftarrow \theta& & \hat \theta&\rightarrow &
				\stpathup{H+\epsilon}{L-\epsilon}{H}{H}& 
			\end{align*}
			where the lower path is selected on the left instance because of the initial observation.
		\end{description}
		This complete the proof of inequality $q_3+q_4 \leq L+H$.
	\end{description}
	We have thus shown that \eqref{eq:sp:characterization:sum-payments:heterogeneous} must also hold in the case the mechanism uses the tie breaking rule in \eqref{eq:sp:heterogeneous:tie-breaking-two}. Thus, from the previously considered tie breaking rule \eqref{eq:sp:heterogeneous:tie-breaking-one}, we conclude that \eqref{eq:sp:characterization:sum-payments:heterogeneous} must hold in any case, which proves the claim.
\end{proof}

\begin{nclaim}
	The constants $q_i$ in Claim~\ref{cl:sp:proof:payments-fixed-plus-served} must be all equal to $\frac{L+H}{2}$. 
\end{nclaim}

\begin{proof}
	By Claim~\ref{cla:sp:characterization:sum-payments:heterogeneous} is is enough to show that $q_1=q_2=q_3=q_4$.
	We first prove $q_1=q_3$. 
	
	By contradiction, if $q_1>q_3$ then we have two possible tie breaking for types $\theta=(L-\epsilon,H+\epsilon,L,H)$:
	\begin{description}
		\item[Case~1.] The mechanism selects the lower path. Since $q_1>q_3$ the coalition given by agents $1$ and $3$ can improve in this scenario:
		\begin{align*}
			\stpathdown{L-\epsilon}{H+\epsilon}{L}{H}& &  
			\leftarrow \theta& & \hat \theta&\rightarrow &
			\stpathup{L-\epsilon}{H+\epsilon}{H}{H}&  
		\end{align*}
		\item[Case~2.] The mechanism selects the upper path.  By Claim~\ref{cla:sp:characterization:sum-payments:heterogeneous}, $q_1>q_3$ implies  $q_2  <  q_4$, and thus  the coalition given by agents $2$ and $4$ can improve in this scenario:
		\begin{align*}
			\stpathup{L-\epsilon}{H+\epsilon}{L}{H}& &  
			\leftarrow \theta& & \hat \theta&\rightarrow &
			\stpathdown{L-\epsilon}{H+\epsilon}{L}{L}&  
		\end{align*}
	\end{description}
	Let us now assume by contradiction that $q_1<q_3$. We consider again the two possible tie breaking for  types $\theta=(H+\epsilon,L-\epsilon,L,H)$:
	\begin{description}
		\item[Case~1.] The mechanism selects the upper path. Since $q_1<q_3$ the coalition given by agents $1$ and $3$ can improve in this scenario:
		\begin{align*}
			\stpathup{H+\epsilon}{L-\epsilon}{H}{L}& &  
			\leftarrow \theta& & \hat \theta&\rightarrow &
			\stpathdown{H+\epsilon}{L-\epsilon}{L}{L}&
		\end{align*}
		\item[Case~2.] The mechanism selects the lower path.  By Claim~\ref{cla:sp:characterization:sum-payments:heterogeneous}, $q_1<q_3$ implies  $q_2  >  q_4$, and thus  the coalition given by agents $2$ and $4$ can improve in this scenario:
		\begin{align*}
			\stpathdown{H+\epsilon}{L-\epsilon}{H}{L}& & 
			\leftarrow \theta& & \hat \theta&\rightarrow &
			\stpathup{H+\epsilon}{L-\epsilon}{H}{H}& 
		\end{align*}
	\end{description}  
\end{proof}
By combining the above claims we have that the payments of a bribeproof mechanism must satisfy \eqref{eq:th:sp:generalized-domain:linear-payments} and thus the theorem follows.

\subsection{Proof of Corollary~\ref{cor:path-auction:no-collusion-proof}}
\begin{proof}
	We show that on the instance in Figure~\ref{fig:sp-diamond} the \linearmechanism{\left(\frac{1}{2}\right)} is not collusion-proof. In particular, we show that even a coalition of three agents can globally obtain a better utility when one of them misreport her type. 
	
	Suppose the (true) types are $\theta = (L,L,L,L)$ and that the mechanism selects the upper path in this case (i.e., agents $1$ and $2$).
	The coalition of agents $2$, $3$ and $4$ can improve their overall utility by reporting $\hat \theta=(L,H,L,L)$:
	\[
	u_2(\theta;L) + u_3(\theta;L) + u_4(\theta;L) = \frac{L+H}{2} + f_2+f_3+f_4 - L 
	\]
	while 
	\[
	u_2(\hat \theta;L) + u_3(\hat \theta;L) + u_4(\hat \theta;L) = 2\frac{L+H}{2} + f_2+f_3+f_4 - L
	\]
	and thus the mechanism cannot be collusion-proof.
\end{proof}

\subsection{Proof of Theorem~\ref{th:sp:no-three-vals}}
\begin{proof}
	Suppose there exists a bribeproof mechanism $(A,p)$ for the path auction on the network in Figure~\ref{fig:sp-diamond} over a three-values domain $\Theta^{(L,M,H)}$, with $L<M<H$. Clearly  $(A,p)$ is bribeproof if we restrict to the two-values domain $\Theta^{(L,M)}$.  Corollary~\ref{cor:sp:characterization} implies that the mechanism must be a \linearmechanism{\left(\frac{1}{2}\right)} for this two-values domain. That is, there exist constants $\{f_i\}_{i\in N}$ such that, for every $\theta \in \Theta^{(L,M)}$ the payments of each agent $i$ are 
	\begin{equation}
	\label{eq:sp-no-three-vals:LM-subdomain}
	p_i(\theta) = f_i + 
	\begin{cases}
	\frac{L+M}{2} & \mbox{if $i$ is selected for types $\theta$,} \\
	\phantom{-}0  & \mbox{otherwise.}
	\end{cases}
	\end{equation}
	To complete the proof we observe that this mechanism is not bribeproof on the original three-values domain $\Theta^{(L,M,H)}$.  Indeed, for types $\theta=(L,H,L,H)$ the two selected agents $i$ and $j$ have total utility
	\[
	f_i+f_j + L + M - (L+H)
	\]
	while if agent $j$ reports $\hat \theta_j=H$ their total utility improves to 
	\[
	f_i+f_j,
	\]
	since in the latter case the path formed by agents $i$ and $j$ does no longer have a minimal cost.
\end{proof}

\subsection{Proof of Theorem~\ref{th:one-job:bribeproof}}

\begin{proof}
	We first prove that the particular tie-breaking rule of the mechanism satisfies the following condition:\footnote{This condition is similar to \emph{replacement monotonicity} \citep[see e.g.][]{Wak13} which dictates that, if an agent changes her own allocation, by changing her reported type, then the allocation of every other agent should change in the opposite direction by a non-larger amount.} 
	For any $i$ and $j\neq i$ the following implication holds
	\begin{align}
		A_i(\theta')&>A_i(\theta'') & \Rightarrow& & A_j(\theta')&\leq A_j(\theta''), \label{eq:replacement-monotonicity}
		\intertext{where $\theta'$ and $\theta''$ are any two type vectors that differ only in $i$'s type,}
		\theta'=&(\theta_i',\theta_{-i}) & \mbox{and}& & \theta''=&(\theta_i'',\theta_{-i}). \nonumber
	\end{align}
	Note that the mechanism breaks ties according to the following strict order:
	\begin{align*}
		\theta_i <_{prec}  \theta_j& & \equiv& & ((\theta_i < \theta_j) \mbox{ or }& (\theta_i=\theta_j \mbox{ and } i < j)).
	\end{align*}
	By contradiction, suppose \eqref{eq:replacement-monotonicity} does not hold, that is\footnote{Note that \eqref{eq:replacement-monotonicity} must hold for $k=1$ item, and thus we are considering $k\geq 2$ for what concerns the proof of \eqref{eq:replacement-monotonicity}.}
	\begin{align*}
		A_i(\theta')=&1=A_j(\theta') & \mbox{and}& & A_i(\theta'')=&0=A_j(\theta'').
		\intertext{Since the number of selected agents is constant, 
			there must be a third agent $h$ such that}
		A_h(\theta')=&0 & \mbox{and}& & A_h(\theta'')=&1.
		\intertext{The allocations for $h$ and $j$ and the tie-breaking rule imply}
		\theta'_j&<_{prec}\theta'_h & \mbox{and}& & \theta'_h=\theta''_h&<_{prec}\theta''_j=\theta'_j
	\end{align*}
	which contradicts the fact that ``$<_{prec}$'' is a strict order.

	We  first show that the strategyproofness holds, that is,
	\[
	u_i(\theta;\theta_i) = (M - \theta_i)A_i(\theta) \geq (M - \theta_i)A_i(\hat \theta_i,\theta_{-i})= u_i(\hat \theta_i,\theta_{-i};\theta_i), 
	\]
	which will follow from the monotonicity of the algorithm. Indeed, since  for $\theta_i =M$ the utilities are equal zero, 
	we only have to consider $\theta_i \in \{L,H\}$. Note that the algorithm satisfies
	\[
	A_i(L,\theta_{-i}) \geq A_i(H,\theta_{-i})
	\]
	because, if $i$ is selected for $\theta_i=H$, then she is also selected for $\theta_i=L$. This implies strategyproofness.
	
	We next prove bribeproofness, that is, that \eqref{eq:bribeproof:breibee} holds for all 
	$\theta$ and $\hat \theta=(\hat \theta_i,\theta_{-i})$, and for all  
	$i$ and $j$. Note that, since we have already proven strategyproofness (i.e., the case $i=j$), we have to consider only the case the utility of $j$ for $\hat \theta$ is be better than her utility for $\theta$, with $j\neq i$. By definition of the payments, this can happen  only in one of these two cases:
	\begin{align*}
		\theta_j=&L, &  A_j(\theta)&=0 & \mbox{and}& & A_j(\hat \theta)&=1; \\
		\theta_j=&H, &  A_j(\theta)&=1 & \mbox{and}& & A_j(\hat \theta)&=0. 
	\end{align*}
	In the first case, condition \eqref{eq:replacement-monotonicity} implies
	\begin{align*}
		A_i(\theta)&=1 & \mbox{and}& & A_i(\hat \theta)&=0,
	\end{align*}
	and thus $\theta_i \leq \theta_j$ follows from $A_i(\theta)=1$ and $A_j(\theta)=0$. Therefore
	\[
	u_i(\theta;\theta_i)+u_j(\theta;\theta_j)= M - \theta_i \geq M - \theta_j = u_i(\hat \theta;\theta_i) + u_j(\hat \theta;\theta_j).
	\]
	In the second case, condition \eqref{eq:replacement-monotonicity} implies
	\begin{align*}
		A_i(\theta)&=0 & \mbox{and}& & A_i(\hat \theta)&=1
	\end{align*}
	and thus $\theta_i \geq \theta_j$ follows from $A_i(\theta)=0$ and $A_j(\theta)=1$. Therefore
	\[
	u_i(\theta;\theta_i)+u_j(\theta;\theta_j)= M - \theta_j \geq M - \theta_i = u_i(\hat \theta;\theta_i) + u_j(\hat \theta;\theta_j).
	\]
	We have thus shown that \eqref{eq:bribeproof:breibee} holds. 
\end{proof}

\subsubsection{Necessity of a particular tie-breaking rule.}
We   note that, because the domain contains three possible values, one must use a particular tie-breaking rule (as the mechanism is Figure~\ref{fig:M-compensation} does). We illustrate this aspect with the following example.

\begin{example}[tie breaking rules]
	Consider the $2$-item auction over three-values domains and with four agents. Suppose the mechanism breaks ties according to the following ordering: 
	\begin{align*}
		(0,0,1,1)\prec (1,1,0,0) \prec \cdots
	\end{align*}
	Then the mechanism which pays each selected agent an amount $M$ (the variant of the mechanism in Figure~\ref{fig:M-compensation} using the above tie-breaking rule) is not bribeproof:
	\begin{align*}
		A(H,L,L,L)&=(0,0,1,1) & \mbox{and}& & A(L,L,L,L)&=(1,1,0,0) 
		\intertext{and, for the case the left type vector corresponds to the true types, the sum of the utilities of the agents $1$ and $2$ in the two cases is}
		0 +& 0 & \mbox{and}& & M -H +& M - L,
	\end{align*}
	meaning that, for any domain with $2M>H+L$, the mechanism violates \eqref{eq:bribeproof:breibee} for $\theta=(H,L,L,L)$ and
	$\hat \theta=(L,L,L,L)$.
\end{example}

\subsection{Proof of Theorem~\ref{th:parallel:three-vals:char}}%
\label{sec:proof:th:parallel:three-vals:char}
We shall prove that the payments of any bribeproof mechanism must satisfy
\begin{align}
\label{eq:fixed-compensation-equivalent:full-proof}
 p_1^{(sel~1)} &= M + p_1^{(sel~2)} & \mbox{ and }& &   
  p_2^{(sel~2)} &= M + p_2^{(sel~1)}\
\end{align}
where $p^{(sel~i)}$ is the payment that the mechanism must choose whenever agent $i$ is selected (we shall see below that 
the payments depend only on the selected agent). This means the mechanism is a \linearmechanism{\lambda_M} with $f_1 = p_1^{(sel~2)}$ and $f_2 = p_2^{(sel~1)}$.

In the following, we depict type vectors  and the corresponding solution together by means of
a gray box indicating the selected agent as follows:
\begin{align}
  \label{eq:opt:allocation}
  (\take{L},M)&, & (\take{L},H)&, & (\take{M},H)& & \mbox{ and }& &
  (M,\take{L})&, & (H,\take{L})&, & (H,\take{M})&.
\end{align}
These instances have a unique optimum solution, while for the case in which the
cost of the two agents is the same, the mechanism can break ties arbitrarily. 
Let us define 
\begin{align}
\label{eq:opt:allocation:payments}
  p^{(sel~1)}:=&p(\take{L},M) & \mbox{ and }& &
  p^{(sel~1)}:=&p(M,\take{L}).
\end{align}
and observe that by Fact~\ref{fac:adj:pay} for the type vectors in \eqref{eq:opt:allocation} 
the payments depend only on the selected agent:
\begin{align}
  \label{eq:opt:allocation:payments-solution-based}
  p(\take{L},M)&=p(\take{L},H)=p(\take{M},H)& & \mbox{ and }& &
  p(M,\take{L})&=p(H,\take{L})=p(H,\take{M})&.
\end{align}
The mechanism, in the remaining three type vectors $(L,L)$, $(M,M)$, $(H,H)$, can break ties in different ways. We show that in either possible case, \eqref{eq:fixed-compensation-equivalent:full-proof} must hold. The proof is broken in the following three lemmas. 

\begin{lemma}\label{le:ties:same}
  If the mechanism breaks ties always in favor of the first agent, i.e.,
  \begin{align}
    (\take{L},L),& & (\take{M},M),& & (\take{H},H)&,
  \end{align}
  then the payments satisfy \eqref{eq:fixed-compensation-equivalent:full-proof}.
\end{lemma}
\begin{proof}
Strategyproofness for agent $1$ yields
\begin{align*}
 \theta&=(\take{M},M) & \hat \theta&=(H,\take{M}) & \Rightarrow& & p_1^{(sel~1)} - M&\geq p_1^{(sel~2)} \\
 \theta&=(M,\take{L}) & \hat \theta&=(\take{L},L) & \Rightarrow& & p_1^{(sel~2)} &\geq
p_1^{(sel~1)} - M
\end{align*}
thus implying 
\[p_1^{(sel~1)}=M + p_1^{(sel~2)}.\]
Similarly, strategyproofness for agent $2$ yields
\begin{align*}
 \theta&=(\take{M},M) & \hat \theta &=(M,\take{L}) & \Rightarrow& & p_2^{(sel~1)}  &\geq p_2^{(sel~2)} - M \\
 \theta&=(H,\take{M}) & \hat \theta &=(\take{H},H) & \Rightarrow& & p_2^{(sel~2)} - M &\geq p_2^{(sel~1)}
  \end{align*}
thus implying \[p_2^{(sel~2)} = M + p_2^{(sel~1)}.\]
\end{proof}

\begin{lemma}\label{le:ties:L-diff-M}
  If the mechanism breaks ties as follows 
  \begin{align}
    (\take{L},L)& &   (M,\take{M})&,
  \end{align}
  then the payments satisfy \eqref{eq:fixed-compensation-equivalent:full-proof}.
\end{lemma}

\begin{proof}
We first prove that the sum of the two payments is constant, i.e., that
\begin{equation}
  \label{eq:one-job:three:ties:equal-sum} 
  p_1^{(sel~1)} + p_2^{(sel~1)} =  p_1^{(sel~2)} + p_2^{(sel~2)}.
\end{equation}
First observe that by Fact~\ref{fac:adj:pay}
\begin{eqnarray*}
 p(\take{L},L)=p(\take{L},M) = p^{(sel~1)}, \\
 p(M,\take{M})=p(H,\take{M}) = p^{(sel~2)}.
\end{eqnarray*}
This and bribeproofness yields: 
\begin{align*}
 \theta&=(\take{L},L) & \hat \theta &=(M,\take{L}) & \Rightarrow& & p_1^{(sel~1)} +
p_2^{(sel~1)} &\geq p_1^{(sel~2)} + p_2^{(sel~2)} \\
 \theta&=(M,\take{M}) & \hat \theta &=(\take{L},M) & \Rightarrow& & p_1^{(sel~2)} +
p_2^{(sel~2)} &\geq p_1^{(sel~1)} + p_2^{(sel~1)}
\intertext{which proves \eqref{eq:one-job:three:ties:equal-sum}. Moreover, strategyproofness for agent $1$ yields}
\theta&=(M,\take{M}) & \hat \theta &=(\take{L},M) & \Rightarrow& & p_1^{(sel~2)} &\geq p_1^{(sel~1)} - M.
  \end{align*}
This and
\eqref{eq:one-job:three:ties:equal-sum} yields $p_2^{(sel~2)} \leq
p_2^{(sel~1)} + M$, while the reversed inequality can be obtained as follows:
\begin{align*}
 \theta&=(M,\take{M}) & \hat \theta &=(\take{M},H) & \Rightarrow& & p_2^{(sel~2)} -M &\geq
p_2^{(sel~1)}  
\end{align*}
Thus $p_2^{(sel~2)} = p_2^{(sel~1)} + M$ and by
\eqref{eq:one-job:three:ties:equal-sum} we obtain  $p_1^{(sel~2)} = p_1^{(sel~1)} -
M$. 
\end{proof}

\begin{lemma}\label{le:ties:M-diff-H}
  If the mechanism breaks ties as follows 
  \begin{align*}
    (\take{L},L)& & (H,\take{H})&,
  \end{align*}
  then the payments satisfy \eqref{eq:fixed-compensation-equivalent:full-proof}.
\end{lemma}

\begin{proof}
We first prove that the sum of the two payments is constant, i.e., that
\begin{equation}
  \label{eq:one-job:three:ties:equal-sum:second-case} 
  p_1^{(sel~1)} + p_2^{(sel~1)} = p_1^{(sel~2)} + p_2^{(sel~2)}.
\end{equation} 
First observe that by Fact~\ref{fac:adj:pay}
\begin{eqnarray*}
 p(\take{L},L)=p(\take{L},M) = p^{(sel~1)}, \\
 p(H,\take{H})=p(H,\take{M}) = p^{(sel~2)}.
\end{eqnarray*}
This and bribeproofness yields: 
\begin{align*}
 \theta&=(\take{L},L) & \hat \theta &=(M,\take{L}) & \Rightarrow& & p_1^{(sel~1)} +
p_2^{(sel~1)} &\geq p_1^{(sel~2)} + p_2^{(sel~2)} \\
 \theta&=(H,\take{H}) & \hat \theta &=(\take{M},H) & \Rightarrow& & p_1^{(sel~1)} +
p_2^{(sel~1)} &\leq p_1^{(sel~2)} + p_2^{(sel~2)}
\intertext{which proves \eqref{eq:one-job:three:ties:equal-sum:second-case}. Moreover, strategyproofness for agent $1$ yields}
\theta &=(M,\take{L}) & \hat\theta&=(\take{L},L)  & \Rightarrow& & p_1^{(sel~2)} &\geq p_1^{(sel~1)} - M \\
\theta &=(\take{M},H) &  \hat  \theta&=(H,\take{H})& \Rightarrow& & p_1^{(sel~1)} - M &\geq p_1^{(sel~2)}
\end{align*}
Thus $p_1^{(sel~1)} - M = p_1^{(sel~2)}$ and by
\eqref{eq:one-job:three:ties:equal-sum:second-case} we  obtain $p_2^{(sel~2)} -
M = p_2^{(sel~1)}$.
\end{proof}

\begin{proof}[Proof of Theorem~\ref{th:parallel:three-vals:char}]
Consider the cases in which both agents have the same type:
\begin{align*}
    (L,L),& & (M,M),& & (H,H)&.
  \end{align*}
  The proof distinguishes among the following cases:
\begin{enumerate}
 \item \emph{The mechanism breaks ties consistently.} Without loss of generality, we can 
 assume ties are broken in favor of agent $1$ and the theorem thus follows from Lemma~\ref{le:ties:same}.
 \item \emph{The mechanism breaks ties between $(L,L)$ and $(M,M)$ in different ways.} We can always rename the agents so that ties are broken as in Lemma~\ref{le:ties:L-diff-M} which then implies the theorem. 
  \item \emph{The mechanism breaks ties between $(L,L)$ and $(H,H)$ in different ways.} We can always rename the agents so that ties are broken as in Lemma~\ref{le:ties:M-diff-H} which then implies the theorem. 
\end{enumerate}
Note that if ties are broken in different ways between $(M,M)$ and $(H,H)$, then on of the last two cases above arises. The theorem  thus follows. 
\end{proof}

\subsection{Proof of Corollary~\ref{cor:parallel:four-vals:char}}
\begin{proof}
	Suppose there exists a bribeproof mechanism $(A,p)$ over a four-values domain $\Theta^{(L,M,M',H)}$. Clearly this mechanism is also bribeproof over the three-values domain
	$\Theta^{(L,M,H)}$. Thus,
	Theorem~\ref{th:parallel:three-vals:char} implies that, for any $\theta\in \Theta^{(L,M,H)}$, the payments must satisfy 
	\begin{equation*}
		p_i(\theta) = f_i + 
		\begin{cases}
			M & \mbox{if $i$ is selected for types $\theta$,} \\
			0  & \mbox{otherwise.}
		\end{cases}
	\end{equation*}
	We shall prove that this mechanism cannot be strategyproof over the original four-values domain since, for $\theta=(M',M',\ldots,M')$, the selected agent $i$ 
	can improve her utility by reporting $\hat \theta_i=H>M'$. Without loss of generality, suppose the selected agent is agent $1$ and observe that (by Fact~\ref{fac:adj:pay}) the payments must satisfy  
	\begin{align*}
		p_1(M',M',M',\ldots,M')&=p_1(L,M',M',\ldots,M') = p_1(L,H,M',\ldots,M')= \cdots  \\
		&= p_1(L,H,\ldots,H) = f_1 + M,
		\intertext{and}
		p_1(H,M',M',\ldots,M')&=p_1(H,M,M',\ldots,M') = \cdots \\
		&= p_1(H,M,M,\ldots,M) = f_1. 
	\end{align*}
	Thus the utility of agent $1$ is $u_1(\theta;\theta_1)=f_1 + M - M'< f_1 = u_1(\hat \theta;\theta_1)$, meaning that the mechanism is not strategyproof.
\end{proof}

\subsection{Proof of Theorem~\ref{th:fractional-makespan}}
\begin{proof}
	Observe that the algorithm minimizing the min-max fairness must produce the following allocation:
	\begin{equation}
	A_k(\theta)= 
	\begin{cases}
	\frac{H}{H n_L + L n_H} & \mbox{if } \theta_k=L, \\
	\frac{L}{H n_L + L n_H} & \mbox{if } \theta_k=H, 
	\end{cases}
	\label{eq:opt-fractional-schedule} 
	\end{equation}
	where $n_L$ and $n_H$ denote the number of agents of type $L$ and $H$ in $\theta$, respectively.
	We next prove  that such an algorithm $A$ satisfying \eqref{eq:opt-fractional-schedule} must satisfy the conditions  of Theorem~\ref{th:sbribeproof:linear}.

	We first prove bounded influence \eqref{eq:bounded-influence}. By rewriting $\influence{k}$ according to \eqref{eq:opt-fractional-schedule} we have
	\begin{align*}
		\influence{i} =&\frac{H}{H n_L + L n_H} - \frac{L}{H n_L - H + L n_H + L}, \\ 
		\influence{\ell} =&\frac{H}{H n_L + L n_H} - \frac{H}{H n_L - H + L n_H + L}, \\
		\influence{h} =&\frac{L}{H n_L + L n_H} - \frac{L}{H n_L - H + L n_H + L}.
	\end{align*}
	Note that all these quantities are nonnegative and thus bounded influence follows from 
	\begin{align*}
		\influence{i} - \influence{\ell}=&\frac{H-L}{H n_L - H + L n_H + L} >0, \\ 
		\influence{i} - \influence{h} =&\frac{H-L}{H n_L + L n_H} >0.
	\end{align*}

	We next show that \eqref{eq:sbribeproof:sufficient:LL-HH} and \eqref{eq:sbribeproof:sufficient:LH-HL} hold. 
	According to \eqref{eq:opt-fractional-schedule} all agents of the same type get the same amount and, in particular, 
	\begin{align*}
		A_i(L,H,\theta_{-ij})&= \frac{H}{L} A_i(H,L,\theta_{-ij}) & \mbox{and}& & A_j(H,L,\theta_{-ij})&= \frac{H}{L} A_j(L,H,\theta_{-ij}), 
	\end{align*}
	which proves \eqref{eq:sbribeproof:sufficient:LH-HL}.
	As for \eqref{eq:sbribeproof:sufficient:LL-HH}, from \eqref{eq:opt-fractional-schedule} 
	\begin{align*}
		A_i(L,L,\theta_{-ij})-A_i(H,H,\theta_{-ij})&= A_j(L,L,\theta_{-ij})-A_j(H,H,\theta_{-ij})\\ 
		& =  \frac{H}{Hn_L + L n_H} - \frac{L}{Hn_L-2H + Ln_L + 2L} >0. 
	\end{align*}
	This completes the proof.
\end{proof}

\subsection{Proof of Theorem~\ref{th:makespan:no-pmon}}
We first show that \emph{exact} solutions cannot be achieved.

\begin{theorem}\label{th-makespan:no-pmon:exact}
	There is no exact bribeproof \linearmechanism{\lambda} for the  makespan minimization on three agents with two values-domains.
\end{theorem}

\begin{proof} 
	We show that the conditions of Theorem~\ref{th:wokload:bribeproof} cannot be satisfied by any $\lambda$ for the instance in Example~\ref{exa:scheduling}. 
	For $i=2$, $\ell =1$ and $h=3$, we have
	\begin{align*}
		\influence{i} &= 10 - 6 & \mbox{and}& & \influence{h} &= - 6 
	\end{align*}
	and thus condition \eqref{eq:pmon:HH} of Theorem~\ref{th:wokload:bribeproof} is violated unless $\lambda=0$. In this case, however, condition \eqref{eq:pmon:HL} is violated for $\ell=1$ since
	\[
	\influence{\ell} = 12 - 10 = 2.
	\] 
	Therefore no exact \linearmechanism{\lambda} can be  bribeproof. 
\end{proof}

We conclude this section by observing that Theorem~\ref{th-makespan:no-pmon:exact} can be strengthen so to include mechanisms that \emph{approximate} the optimal makespan. We say that a mechanism approximates the makespan within a factor $\alpha$ if 
its algorithm returns always an allocation whose makespan is at most $\alpha$ times the minimal makespan.

\begin{theorem}[Theorem~\ref{th:makespan:no-pmon} in main text]\label{cor-makespan:no-pmon:inapx}
	No bribeproof \linearmechanism{\lambda} for the  makespan minimization on three agents with two values-domains can approximate the makespan within a factor smaller than $\valrho \approx 1.1547$.
\end{theorem}

\begin{proof}
	We prove the result by using the same instances used to prove Theorem~\ref{th-makespan:no-pmon:exact}, in which we set the 
	parameter $\epsilon$ (recall that $L=1$ and $H=2+\epsilon$) so to obtain the best lower bound. Since no bribeproof \linearmechanism{\lambda} can return the optimum in both type vectors $\theta=(L,L,H)$ and $\hat \theta=(L,H,H)$, it must return a suboptimal solution in at least one case:
	\begin{enumerate}
		\item If the mechanism does not return the optimum in the first instance, then agent of type $H$ must get at least one job of size $6$, and the resulting makespan will be $6H=12 + 6\epsilon$ (or larger). That is, the  mechanism does not approximate the makespan within a factor better than  $\alpha_1:=\frac{12+6\epsilon}{12}$.
		\item If the mechanism does not return the optimum in the second instance, then  the best alternative solution is to allocate the job of weight $10$ together with a job of weight $6$ to the agent of type $L$, and the resulting type will be $16L=16$. That is, the  mechanism does not approximate the makespan within a factor better than  $\alpha_2:=\frac{16}{12+6\epsilon}$.
	\end{enumerate}
	Thus the mechanism cannot approximate the makespan within a factor better than  $\alpha=\min\{\alpha_1,\alpha_2\}$. For obtaining the best lower bound we choose $\epsilon$ so to maximize this quantity: For $\epsilon=\valepsilon$ we have $\alpha_1=\alpha_2=1 + \epsilon/2 = \valrho \approx 1.1547.$
\end{proof}

\subsection{Proof of Theorem~\ref{th:sp:no-three-vals:triangle}}
\begin{proof}
Consider the network in Figure~\ref{fig:sp:no-three-vals-sp:triangle} along with 
and a tree-values domain $\Theta^{(L,M,H)}$ satisfying 
\begin{align*}
 2L &< M, &  L+H &< 2M, & \mbox{and}& & L+M &<
H.
\end{align*}
(For instance, we can take $L=1$, $M= 4$ and $H=6$.) Consider the following
two type vectors along with the (unique) path of minimal cost
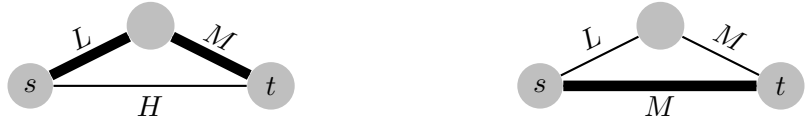
\begin{align*}
	&\stpathtriangleup{L}{M}{H} & &\stpathtriangledown{L}{M}{M}
\end{align*}
and let $p^{(top)}$ and $p^{(bot)}$ denote the corresponding payments,
\begin{align*}
 p^{(top)} &:= p(L,M,H) & \mbox{and}& & p^{(bot)}&:= p(L,M,M).
\end{align*}
(These payments correspond to the case the mechanism selects the top path and the bottom
path for these two type vectors). We shall prove below that bribeproofness forces the following two conditions:
\begin{eqnarray}
  p_3^{(bot)} - p_3^{(top)} & = &  M, \label{eq:no-three-vals-sp:payments-down}\\
  p_2^{(top)}  -  p_2^{(bot)} & = & L.
 \label{eq:no-three-vals-sp:payments-up}
\end{eqnarray}
%
From this we shall conclude that bribeproofness is violated by considering $\theta = (L,M,H)$ and $\hat \theta = (L,M,M)$. 
In this case, 
\begin{eqnarray*}
  u_2(\theta;\theta_2) - u_2(\hat \theta;\theta_2)  & = & p_2^{(top)}  - M - p_2^{(bot)} = L - M\\ 
  u_3(\theta;\theta_2) - u_3(\hat \theta; \theta_2)  & = & p_3^{(top)} - p_3^{(bot)} + H = H - M
\end{eqnarray*}
By the assumption $L+H < 2M$ the above two identities imply
\begin{align*}
 u_2(\theta;\theta_2) + u_3(\theta;\theta_3) < u_2(\hat \theta;\theta_2) + u_3(\hat \theta;\theta_3)
\end{align*}
which means that the mechanism is
not bribeproof.

To complete the proof, we make use of the following two additional type vectors along with the (unique) path of minimal cost:
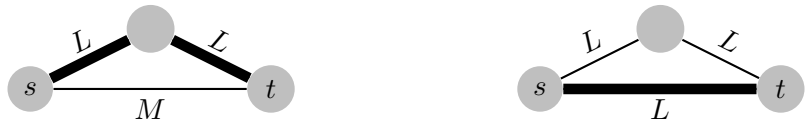
\begin{align*}
	&\stpathtriangleup{L}{L}{M} & &\stpathtriangledown{L}{L}{L}
\end{align*}
We first show that also in this case the payments are $p^{(top)}$ and $p^{(bot)}$ defined above, that is
\begin{align*}
 p(L,L,M) &=p^{(top)}  & \mbox{and}& & p(L,L,L)&=p^{(bot)}.
\end{align*}
To this end, we apply 
Fact~\ref{fac:adj:pay} in each case:
\begin{enumerate}
  \item Type vector $(L,M,H)$ is $A$-equivalent to $(L,L,H)$ and the latter is $A$-equivalent to
     $(L,L,M)$. Fact~\ref{fac:adj:pay} implies $p(L,M,H) = p(L,L,M)$, that is $p(L,L,M) =p^{(top)}$.
  \item Type vector $(L,M,M)$ is $A$-equivalent to  $(L,M,L)$  and the latter is $A$-equivalent to
    $(L,L,L)$. Fact~\ref{fac:adj:pay} implies $p(L,M,M) = p(L,L,L)$, that is $p(L,L,L) =p^{(bot)}$. 
\end{enumerate}

\paragraph{Proof of \eqref{eq:no-three-vals-sp:payments-down}.}
Strategyproofness for agent $3$ requires the following two conditions:
\begin{align*}
 p_3(L,M,M) - M &\geq p_3(L,M,H) & \mbox{and}& &
 p_3(L,L,M)  &\geq p_3(L,L,L) - M, \\
 \intertext{that is}
 p_3^{(bot)} - M &\geq p_3^{(top)} & \mbox{and}& &
 p_3^{(top)}  &\geq p_3^{(bot)} - M,
\end{align*}
which clearly imply \eqref{eq:no-three-vals-sp:payments-down}.

\paragraph{Proof of \eqref{eq:no-three-vals-sp:payments-up}.} 
We first observe that
\[
  p(M,L,M) = p(M,M,M) = p(L,M,M) = p^{(bot)}
\]
because these are three $A$-equivalent type vectors and thus Fact~\ref{fac:adj:pay} applies.

Suppose first that $p_2^{(top)} - p_2^{(bot)} > L$. We show that in this case bribeproofness is violated for 
$\theta=(M,L,M)$ and $\hat \theta=(M,L,H)$ since
\begin{align}
 u_2(M,L,M;L) + u_3(M,L,M;M) &= p_2(M,L,M) + p_3(M,L,M) - M  \nonumber \\
    &= p_2^{(bot)} + p_3^{(bot)} - M, \label{eq:no-three-vals-sp:payments-up:true} 
 \intertext{and}
 u_2(M,L,H;L) + u_3(M,L,H;M) &= p_2(M,L,H) - L + p_3(M,L,H)  \nonumber \\
    &= p_2^{(top)} - L + p_3^{(top)}. \label{eq:no-three-vals-sp:payments-up:false}
\end{align}
By \eqref{eq:no-three-vals-sp:payments-down} and by the hypothesis $p_2^{(top)} - p_2^{(bot)} > L$, the quantity in \eqref{eq:no-three-vals-sp:payments-up:false} is larger than the quantity in \eqref{eq:no-three-vals-sp:payments-up:true}, which clearly violates bribeproofness.

Similarly, if $p_2^{(top)} - p_2^{(bot)} < L$ then  bribeproofness is violated for 
$\theta=(L,L,M)$ and $\hat \theta=(L,L,L)$ since
\begin{align}
 u_2(L,L,M;L) + u_3(L,L,M;M) &= p_2(L,L,M) - L + p_3(L,L,M)  \nonumber \\
    &= p_2^{(top)} -L + p_3^{(top)}, 
 \intertext{and}
 u_2(L,L,L;L) + u_3(L,L,L;M) &= p_2(L,L,L) + p_3(L,L,L) - M  \nonumber \\
    &= p_2^{(bot)} + p_3^{(bot)} - M. 
\end{align}
Again, by \eqref{eq:no-three-vals-sp:payments-down} and by the hypothesis $p_2^{(top)} - p_2^{(bot)} < L$, the latter quantity is larger and thus bribeproofness is violated.
\end{proof}

\end{document}